\documentclass[11pt]{article}
\usepackage{graphicx}
\usepackage[T1]{fontenc}

\usepackage{epsfig}
\usepackage[english]{babel}
\usepackage{amsfonts,amsmath,enumerate}
\usepackage{amsmath,amssymb}
\usepackage{amscd}
\usepackage{amsthm}
\usepackage{amsmath,amsfonts,amssymb,mathrsfs}
\usepackage{latexsym}
\usepackage{stmaryrd}
\usepackage{amssymb}
\usepackage{color}

\usepackage[round]{natbib}
\usepackage[titletoc]{appendix}

\newlength{\oldparindent}
\usepackage{fancyhdr}
\usepackage{hyperref} 
\hypersetup{
	colorlinks = true,
	linkcolor = blue,
	anchorcolor = blue,
	citecolor = blue,
	filecolor = blue,
	urlcolor = blue
}

\newtheorem{theorem}{Theorem}[section]
\newtheorem{corollary}[theorem]{Corollary}

\newtheorem{proposition}[theorem]{Proposition}
\newtheorem{example}[theorem]{Example}
\newtheorem{definition}[theorem]{Definition}
\newtheorem{remark}[theorem]{Remark}
\newtheorem{problem}[theorem]{Problem}
\newtheorem{assumption}[theorem]{Assumption}

\title{The Entropic Measure Transform} %

\title{\Large \bf The Entropic Measure Transform} %

\author{
  \textsc{Renjie WANG}\footnotemark[1]
\quad
  \textsc{Cody HYNDMAN}\thanks{Department of Mathematics and Statistics, Concordia University, 1455 Boulevard de Maisonneuve Ouest, Montr\'{e}al, Qu\'{e}bec, Canada H3G 1M8. Email: cody.hyndman@concordia.ca}
\quad
\textsc{Anastasis KRATSIOS}\thanks{Department of Mathematics, ETH Z\"{u}rich, Z\"{u}rich, Switzerland. ORCID ID: 0000-0001-6791-3371. Email: anastasis.kratsios@math.ethz.ch}
}

\date{%
February 20, 2019 \\ %
}

\begin{document}
\newpage

\maketitle
\numberwithin{equation}{section}

\begin{abstract}
We introduce the entropic measure transform (EMT) problem for a general process and prove the existence of a unique optimal measure characterizing the solution.  The density process of the optimal measure is characterized using a semimartingale BSDE under general conditions.   The EMT is used to reinterpret the conditional entropic risk-measure and to obtain a convenient formula for the conditional expectation of a process which admits an affine representation under a related measure.  The entropic measure transform  is then used provide a new characterization of defaultable bond prices, forward prices, and futures prices when the asset is driven by a jump diffusion.  The characterization of these pricing problems in terms of the EMT provides economic interpretations as a maximization of returns subject to a penalty for removing financial risk as expressed through the aggregate relative entropy. The EMT is shown to extend the optimal stochastic control characterization of default-free bond prices of Gombani and Runggaldier (Math. Financ. 23(4):659-686, 2013).  These methods are illustrated numerically with an example in the defaultable bond setting.  
\end{abstract}


\noindent
%{\bf Acknowledgements:}
\let\thefootnote\relax\footnotetext{This research was supported by the Natural Sciences and Engineering Research Council of Canada (NSERC).  The authors would like to thank Professor W.\ Runggaldier (Padova) for helpful comments on an earlier version of this paper.}


\newpage

\section{Introduction}
 The  pricing problem for  zero-coupon bonds based on an underlying short term interest rate process $r(t)\in \mathbb{R}^{+}$ is a fundamental and important topic in financial mathematics. Various models for $r(t)$ have been proposed under the risk neutral measure. One-factor models use the instantaneous spot rate $r(t)$ as the basic state variable, such as \citet{vasicek1977equilibrium} and  \citet{cox1985theory}. Multi-factor models in which the short rate depends on a multidimensional factor process include the models of \citet{longstaff1992interest},   \citet{hull1994numerical}, and \citet{duffie1996yield}. There are several ways to characterize the bond price.  In an arbitrage free market the bond price can be viewed as a solution to a partial differential equation called the term-structure equation (see \citet[Proposition 21.2]{bjork2004arbitrage}) or, linked through the Feynman-Kac formula, by using risk neutral valuation (see \citet[Proposition 21.3]{bjork2004arbitrage}).
 Recently alternative  approaches have been studied including the stochastic flow approach (see \citet{Elliott},   \citet{Zhou2010}, and \citet{hyndman2009forward}), a forward-backward stochastic differential equation approach (see \citet{hyndman2007forward,hyndman2009forward} and \citet{Zhou2010}), and an optimal stochastic control approach of  \citet{gombani2012arbitrage}.
 
\citet{gombani2012arbitrage} associate the pricing problem of default-free bonds with an optimal stochastic control (OSC) problem by transforming the term-structure equation to an equivalent Hamilton-Jacobi-Bellman equation. Inspired by \citet{gombani2012arbitrage} and the notion of  relative entropy we develop an entropic measure transform (EMT) problem whose value function is connected with the price of bonds. We explore the equivalence between the EMT problem and OSC problem.  One advantage of the EMT problem compared to the OSC problem is the straightforward extension to models with jumps or even to models for defaultable bonds.  The EMT problem also provides a financial interpretation of the pricing problem in terms of maximization of returns subject to an entropy penalty term that quantifies financial risk.

We show that the optimal measure and the value process of the EMT problem can be completely characterized by a forward-backward stochastic differential equation (FBSDE).  In addition, the entropic measure transform has an explicit expression provided that the related FBSDE admits an explicit solution.  From the explicit representation of the entropic measure transform we note that the  measure which solves the EMT problem coincides with the martingale measure using the bond price as num\'{e}raire, or the forward measure. These connections provide some insight into why the forward measure transformation employed in the FBSDE approach of \citet{hyndman2009forward} is effective. Under the framework of affine term-structure models (ATSMs) and QTSMs, \citet{hyndman2009forward} and \citet{Zhou2010} presented explicit solutions for the related FBSDE.

The remainder of the paper is organized as follows. In Section \ref{s_gen_EMT}, the entropic measure transform (EMT) problem is introduced in full generality, solved, and characterized in terms of a backwards semimartingale.  
In Section~\ref{s_gen_connections}, connections between the EMT, the conditional entropic risk-measure, and affine processes are made.  Using the EMT, it is shown that the entropic measure transform can be used to simply compute the conditional expectation of a stochastic process which can be represented as an affine process under a specific equivalent measure.  
In Section \ref{s_EMT_if}, the entropic measure transform is applied to pricing problems, beginning with default-free zero-coupon bonds, then extended to defaultable zero-coupon bonds, and lastly it is applied to futures and forward prices.  In all these cases, the backwards semimartingale characterizing the optimal measure is reduced to an FBSDE whose solution is given in terms of a Riccati-equation. We also establish an equivalence relation between the OSC problem and the EMT problems for bond pricing. Section~\ref{s_Numerical} contains a numerical illustration of the implementing the method in the case of defaultable bonds.    
Section \ref{conclusion} concludes and an appendix discusses the solvability of certain Riccati equations.

\section{The Entropic Measure Transform}\label{s_gen_EMT}
In this section, we introduce the entropic measure transform of a probability measure, and describe how it may be computed.  Let $\mathcal{P}(\Omega)$ be the set of probability measures absolutely continuous with respect to
$\mathbb{P}$  on  $(\Omega, \mathcal{F})$.  The following definitions generalize the classic definitions of 
the free energy and the relative entropy given in \citet{dai1996connections} to the aggregate or dynamic 
version that incorporates the presence of a filtration $\mathcal{F}_t$.  

\begin{definition}\label{def:free:energy:a}
For $\mathbb{P}\in \mathcal{P}(\Omega)$ and $\varphi$  an $\mathcal{F}_T$-measurable random variable the 
aggregate free energy of $\varphi$ with respect to $\mathbb{P}$, 	$\varepsilon_{t,T}(\varphi)$,  is defined 
by
\begin{equation}\label{eq:free:energy:a}
	\varepsilon_{t,T}(\varphi)=\ln(\mathbb{E}_{\mathbb{P}}[e^{\varphi}|\mathcal{F}_t]),\quad t\in[0,T].  
\end{equation} 
\end{definition}

\begin{definition}\label{def:reletive:entropy:a}
	Consider, in addition to $\mathbb{P}$, another $\mathbb{Q} \in \mathcal{P}(\Omega)$. Suppose the
	Radon-Nikodym derivative of $\mathbb{Q}$ with respect to $\mathbb{P}$ is 
	\begin{equation}\label{eq:reletive:entropy:a}
		\left.\frac{d \mathbb{Q}}{d \mathbb{P}}\right|_{\mathcal{F}_s}=\Gamma_s,\quad  0\leq s \leq T.
	\end{equation} 
	Then, for $t\in [0,T]$,	the aggregate relative entropy of $\mathbb{Q}$ with respect to $\mathbb{P}$ is 
	defined as
	\begin{equation}\label{eq:reletive:entropy:b}
		H_{t,T}(\mathbb{Q}|\mathbb{P}) =
		\begin{cases}
			\mathbb{E}_\mathbb{Q}\left[\ln\left(\frac{\Gamma_T}{ \Gamma_t}\right)|\mathcal{F}_t\right] & \text{if}\quad \ln\left( \frac{\Gamma_T}{\Gamma_s}\right)\in L^1(\mathbb{P}), \\
			+\infty       & \text{otherwise}.
		\end{cases}
	\end{equation}
\end{definition}

The duality relation between the aggregate free energy and aggregate relative entropy relies on the following set of measures.  
For $t\in [0,T]$, and for every $\mathbb{P}$-a.s.\ positive and uniformly integrable $(\mathbb{P},\mathcal{F}_t)$-martingale $\Lambda_t$, satisfying
\begin{enumerate}[(i)]
	\item
	$
	\mathbb{E}\left[
	\lim\limits_{t\mapsto \infty}\Lambda_t
	\right]=1,
	$
%	\item The measure $\tilde{\mathbb{P}}$ satisfies $\frac{d\tilde{\mathbb{P}}}{d\mathbb{P}}= \Lambda_s$ for $0\leq s\leq t$,
	\item The measure $\tilde{\mathbb{P}}$ defined by $\frac{d\tilde{\mathbb{P}}}{d\mathbb{P}}=\lim\limits_{t\mapsto \infty}\Lambda_t$ is equivalent to $\mathbb{P}$.  
\end{enumerate}
We may define a family of probability measures $\mathcal{P}_t(\Lambda) \subseteq \mathcal{P}(\Omega)$ which are indistinguishable from $\mathbb{P}$ up to time $t$ by
\begin{equation}\label{eq:measure:subset:a}
	\mathcal{P}_t(\Lambda)=\left\{
	\mathbb{Q}\in \mathcal{P}(\Omega)~\Big|~\mathbb{Q}\sim \mathbb{P},%
	\frac{d\mathbb{P}}{d\mathbb{Q}}\Big|_{\mathcal{F}_t}=\left(\frac{d\mathbb{Q}}{d\mathbb{P}}
	\Big|_{\mathcal{F}_t}\right)^{-1}
	(\forall t \in [0,\infty))
	\text{ and }%
	\frac{d\mathbb{Q}}{d\mathbb{P}}\Big|_{\mathcal{F}_s}=\Lambda_s
	(\forall s \in [0,t])
	\right\}.
\end{equation}

\begin{definition}[$\phi$-Compatible]\label{ass_nice_behave}
	Given an $\mathcal{F}_T$-measurable random-variable $\phi$, a $\mathcal{F}_t$-predictable process $\theta_t$ is said to be $\phi$-compatible if satisfies:
	\begin{enumerate}
\item $\theta_t$ is $\mathbb{P}$-a.s.\ positive,
\item $\mathbb{E}_{\mathbb{P}}\left[e^{\theta_t \phi}\middle| \mathcal{F}_t\right]$ is a uniformly-integrable $(\mathbb{P},\mathcal{F}_t)$-martingale, for every $t \in [0,T]$.  
	\end{enumerate}
\end{definition}

Similar to \citet{dai1996connections}, the following proposition reveals the duality relationship between the aggregate free energy and the aggregate relative entropy.  

\begin{proposition}\label{proposition:duality:relation:a}
	For $t\in [0,T]$ and  any $\mathcal{F}_T$-measurable random variable $\phi$ and any non-negative $\mathcal{F}_t$-predictable processes $\theta_t,\alpha_t$, if $\theta_t$ is $\phi$-compatible, then the following holds
	\begin{equation}\label{eq:duality:relation:a}
		-\varepsilon_{t,T}(\theta_t\phi -\alpha_t + \ln\left(\Lambda_t\right))
		=
		\theta_t
		\inf_{\mathbb{Q}\in
			\mathcal{P}_t(\Lambda)}\left\{\mathbb{E}_{\mathbb{Q}}\left[-\phi+\frac{\alpha_t}{\theta_t}\middle |\mathcal{F}_t\right]+\frac1{\theta_t}
		H_{t,T}(\mathbb{Q}|\mathbb{P})\right\}
		. 
	\end{equation}
	The unique infimum is attained at $\mathbb{P}^{\star}$ determined by the Radon-Nikodym derivative  
	\begin{equation}\label{eq:optimal:measure:a}
		\left.	\frac{d \mathbb{P}^\star}{d \mathbb{P}}\right|_{\mathcal{F}_T}=
		\frac{e^{\theta_t\phi}
	}{\mathbb{E}_{\mathbb{P}}[e^{\theta_t\phi}|\mathcal{F}_t]}.  
	\end{equation}
\end{proposition}
\begin{proof}
	First assume that $\theta_t=1$ $\mathbb{P}$-a.s.  As in equation (\ref{eq:reletive:entropy:a}) we suppose 
	\begin{equation*}
		\left.\frac{d \mathbb{Q}}{d \mathbb{P}}\right|_{\mathcal{F}_s}=\Gamma_s,\quad  0\leq s \leq T.  
	\end{equation*}
	Since $\phi \cdot \frac{\Gamma_t}{\Gamma_T}$ is $\mathcal{F}_T$-measurable then~\citep[Theorem 3.8]{JacodShiryaevLimitTheorems} and the assumption that the following generalized Bayes' formula holds
	$
	\frac{d\mathbb{P}}{d\mathbb{Q}}\Big|_{\mathcal{F}_t}=\left(\frac{d\mathbb{Q}}{d\mathbb{P}}
	\Big|_{\mathcal{F}_t}\right)^{-1}
	$ for every $t \geq 0$, 
	imply that
	\begin{equation}
		\mathbb{E}_{\mathbb{Q}}\left[\phi \frac{\Gamma_t}{\Gamma_T}\middle|\mathcal{F}_t\right]
		=
		\mathbb{E}_{\mathbb{P}}\left[\phi \frac{\Gamma_t}{\Gamma_T}\frac{\Gamma_T}{\Gamma_t}\middle|\mathcal{F}_t\right]
		=\mathbb{E}_{\mathbb{P}}\left[\phi\middle|\mathcal{F}_t
		\right].
		\label{eq_reverse_abs_bayes_a}
	\end{equation}	
	Hence for any $\mathcal{F}_T$-measurable random-variable $\phi$ and any $\mathbb{Q}\in\mathcal{P}_t(\Lambda)$ the following reverse-abstract Bayes' formula holds
	\begin{equation}
		\mathbb{E}_{\mathbb{P}}\left[\phi\middle|\mathcal{F}_t\right]=\mathbb{E}_{\mathbb{Q}}\left[\phi\frac{\Gamma_t}{\Gamma_T}\middle|\mathscr{F}_t\right]
		.
		\label{eq_reverse_abs_bayes_b}
	\end{equation}	

	Since $-\ln (\cdot)$ is a convex function, Jensen's inequality implies that
	\begin{align}
		\nonumber
		-\varepsilon_{t,T}(\phi)=-\ln(\mathbb{E}_{\mathbb{Q}}\Big[\frac{e^{\phi}\Gamma_t}{\Gamma_T}\Big|\mathcal{F}_t\Big]) 
		\leq& 
		\mathbb{E}_{\mathbb{Q}}[-
		\phi|\mathcal{F}_t]+\mathbb{E}_{\mathbb{Q}}\Big[\ln\left(\frac{\Gamma_T}{\Gamma_t}\right)\Big|\mathcal{F}_t\Big]\\
		=&
		\mathbb{E}_{\mathbb{Q}}[-\phi|\mathcal{F}_t]+H_{t,T}(\mathbb{Q}|\mathbb{P})
		. 	
		\label{eq:duality:relation:proof:c} 
	\end{align}
	The ansatz,
	\begin{equation}
		\frac{\Gamma_T}{\Gamma_t}=\frac{e^{\phi}}{\mathbb{E}_{\mathbb{P}}[e^{\phi}|\mathcal{F}_t]}.  
		\label{eq_ansatz}
	\end{equation}
	can be verified to respect equation (\ref{eq:duality:relation:proof:c}) with equality. 
	Since $\Gamma_s=\Lambda_s$ for every $s\leq t$ then equation~\eqref{eq_ansatz} implies that
	\begin{equation*}
		\Gamma_T=\frac{e^{\phi}\Lambda_t}{\mathbb{E}_{\mathbb{P}}[e^{\phi}|\mathcal{F}_t]}.  
	\end{equation*}
	This establishes the case where $\theta_t=1$.  
	For the general case, let $\varphi$ be a $\mathscr{F}_T$-measurable random variable.  Then $\phi\triangleq \theta_t\varphi$ is $\mathscr{F}_T$-measurable.  Therefore, the first part implies that
	\begin{align}
		-\varepsilon_{t,T}(\phi)=	-\varepsilon_{t,T}(\theta_t\varphi -\alpha_t)
		=&
		\inf_{\mathbb{Q}\in
			\mathcal{P}_t(\Lambda)}\{\mathbb{E}_{\mathbb{Q}}[-\theta_t\varphi|\mathcal{F}_t]+
		H_{t,T}(\mathbb{Q}|\mathbb{P})\}\\
		=&
		\inf_{\mathbb{Q}\in
			\mathcal{P}_t(\Lambda)}
		\theta_t
		\{\mathbb{E}_{\mathbb{Q}}[-\varphi + \frac{\alpha_t}{\theta_t}|\mathcal{F}_t]+
		\frac1{\theta_t}H_{t,T}(\mathbb{Q}|\mathbb{P})\}
		,
	\end{align}
	where the last line follows from the fact that $\theta_t$ and $\alpha_t$ are $\mathcal{F}_t$-predictable and do not enter into the optimization.  Moreover, division by $\theta_t$ is well-defined since it $\mathbb{P}$-a.s.\ takes vales in $(0,\infty)$.  Lastly, leaving $\Lambda_t$ inside the expectation and noting that $\ln(\Lambda_t)$ is $\mathbb{P}$-a.s.\ well-defined since $\Lambda_t>0$,  $\mathbb{P}-a.s.$\ yields the existence as well as equation~\eqref{eq:optimal:measure:a}.  
	The process defined in equation~\eqref{eq:optimal:measure:a} is indeed a well-defined density due to the assumption on $\theta_t$ stated in Definition~\ref{ass_nice_behave}(i-ii) and as $s$ approaches $T$ from the left, then 
	$\frac{
		\mathbb{E}_{\mathbb{P}}\left[e^{\theta_s\varphi}\middle|\mathcal{F}_s\right]
	}{\mathbb{E}_{\mathbb{P}}[e^{\theta_t\varphi}|\mathcal{F}_t]}$ converges to the right-hand side of equation~\eqref{eq:optimal:measure:a}.  The uniqueness is guaranteed by the convexity of the loss function $\theta_t
	\{\mathbb{E}_{\mathbb{Q}}[-\varphi + \frac{\alpha_t}{\theta_t}|\mathcal{F}_t]+
	\frac1{\theta_t}H_{t,T}(\mathbb{Q}|\mathbb{P})\}$, with respect to $\mathbb{Q}$.  
\end{proof}

\begin{definition}[Entropic Measure Transform]\label{defn_EMT}
	The function from $\mathcal{P}(\Omega)$ to itself, mapping any probability measure $\mathbb{P}$ to the measure $\mathbb{P}^{\star}$ defined as the (unique) minimizer of equation~\eqref{eq:duality:relation:a}, is called the entropic measure transform.  
\end{definition}

The density process $\frac{d\mathbb{P}^{\star}}{d\mathbb{P}}|_{\mathcal{F}_t}$ admits a convenient description in terms of a decoupled forward-backwards semimartingale.  We review a few definitions and fix some notation in order to state the next and central theorem of the paper.  Let $\mathcal{M}^{\mathbb{P}}$ denote the set of square-integrable $(\mathcal{F}_t,\mathbb{P})$-martingales with initial state equal to $0$.  For any subset $A$ of $\mathcal{M}^{\mathbb{P}}$, the stable subspace, denoted by $\mathcal{S}(A)$ is the smallest linear subspace of $\mathcal{M}$ containing $A$ and satisfying
$$
\left(\forall \tilde{M}\in \mathcal{S}(A)\right)
\quad
\tilde{M} \in \mathcal{S}(A) \Rightarrow \int \phi d\tilde{M}
.
$$ Moreover, we will denote by $\mathbb{M}(A)$, the set of probability measures on $\left(\Omega,\mathcal{F}_T\right)$ for which each $\tilde{M}\in \mathcal{S}(A)$ is a square-integrable martingale.  

\begin{theorem}\label{thrm_main}
	Let $\mathbb{P}$ be a probability measure in $\mathcal{P}_t(1)$, $M_1,\dots,M_n$ be an orthogonal set of $(\mathbb{P},\mathcal{F}_t)$-martingales, and $X_T$ be a $(\mathbb{P},\mathcal{F}_T)$-square integrable random-variable.  
	If $\,\mathbb{P}$ is an extremal point of $\mathcal{M}\left(M_1,\dots,M_n\right)$, if $\mathcal{S}\left(M_1,\dots,M_n\right)$ contains all the constants, and if for every $s \in [t,T]$, there exists a solution to the following semimartingale backwards stochastic differential equation
	\begin{align}
	\nonumber
	-\ln\left(
	\mathbb{E}_{\mathbb{P}}\left[
	e^{X_T}
	\middle| \mathcal{F}_s
	\right]
	\right)  = & {X_T}
	-
	\sum_{i=1}^{n}\int^T_s Z_{u_-}^idM^i_u 
	+
	\sum_{i=1}^{n}\frac1{2}
	\int^T_s(Z_s^i)^2
	d[M^i]_s
	\\
	-&
	\sum_{u\geq s}^{T}
	\left(
	\Delta \mathbb{E}_{\mathbb{P}}\left[e^{X_T}
	\middle| \mathcal{F}_u
	\right]
	-
	\sum_{i=1}^{n}Z_{u-}^i\Delta M^i_u
	\right)
	,
	\label{thrm_main_defining_system}
	\end{align}
	then the density process measure $\frac{d\mathbb{P}^{\star}}{d\mathbb{P}}\Big|_{\mathcal{F}_s}$, is characterized by
	\begin{align}\nonumber
	\frac{d\mathbb{P}^{\star}}{d\mathbb{P}}\Big|_{\mathcal{F}_s}= & \exp\left(
	-
	\sum_{i=1}^{n}\int^T_s Z_{u_-}^idM^i_u 
	+
	\sum_{i=1}^{n}\frac1{2}
	\int^T_s(Z_s^i)^2
	d[M^i]_s\right)\times\\
	&\exp\left(-
	\sum_{u\geq s}^{T}
	\left(
		\Delta \mathbb{E}_{\mathbb{P}}\left[e^{X_T}
		\middle| \mathcal{F}_u
		\right]
	-
	\sum_{i=1}^{n}Z_{u-}^i\Delta M^i_u
	\right)
	\right)
\label{thrm_main_defining_equation}
.
	\end{align}
\end{theorem}
\begin{proof}
For $s \in [t,T]$, let $-Y_s\triangleq \ln\left(
\mathbb{E}_{\mathbb{P}}\left[
e^{X_T}
\middle| \mathcal{F}_s
\right]
\right)$, then 
\begin{equation}
e^{-Y_s} = \mathbb{E}_{\mathbb{P}}\left[
e^{X_T}
\middle| \mathcal{F}_s
\right]
.
\label{eq_thrm_1}
\end{equation}
Define $\eta_s \triangleq \mathbb{E}_{\mathbb{P}}\left[
e^{X_T}
\middle| \mathcal{F}_s
\right]$; notice that $\eta_s$ is the Doob $(\mathbb{P},\mathcal{F}_s)$-martingale.  Since $\mathbb{P}$ is extremal on $\mathbb{M}\left(M_1,\dots,M_n\right)$ and $\mathcal{S}\left(M_1,\dots,M_n\right)$ contains all the constants, then the central theorem of \cite{JacodYorThrm} implies that $\mathcal{M}^{\mathbb{P}}=\mathcal{S}\left(M_1,\dots,M_n\right)$.  In particular, $\left\{M_i\right\}_{i=1}^n$ have the predictable representation property.  Hence, there exists $\mathbb{P}$-a.s.\ unique square-integrable processes $\phi^i_s$ with each $\phi^i_s$ being $M_s^i$-predictable, such that 
\begin{equation}
\eta_s 
	= 
\eta_0 + \sum_{i=1}^n \int_0^s
	\phi_u^i
dM_u^i
.
\label{eq_thrm_2}
\end{equation}
Since $e^{-X_T}$ is $\mathbb{P}$-a.s.\ positive, then $\eta_s$ is $\mathbb{P}$-a.s.\ strictly positive.  Therefore, for each $i=1,\dots,n$, the process $Z_s^i\triangleq \frac{\phi^i_s}{\eta_s}$ is well-defined.  Therefore, equation~\eqref{eq_thrm_2} may be rewritten as
\begin{equation}
\eta_s 
= 
\eta_0 + \sum_{i=1}^n \int_0^s
\eta_s Z_s^i
dM_u^i
.
\label{eq_thrm_3}
\end{equation}

Since $Y_s = -\ln(\eta_s)$, then the generalized It\^{o} Lemma, \citep[Theorem 14.2.1]{SamElliotStochastics}, equation~\eqref{eq_thrm_3}, and the orthogonality of $M^i$ to $M^j$ for $i\neq j$, imply that
\begin{align}
\nonumber
Y_s 		=& 
Y_0
	+ 
\sum_{i=1}^{n}\int_0^s \frac{\eta_{u_-}Z_{u_-}^i}{\eta_{u_-}}dM^i 
	- 
\sum_{i=1}^{n}\frac1{2}
			\int_0^s
				\frac{\eta_s^2 (Z_s^i)^2 }{\eta_s^2}
			d[M^i]_s
	-\sum_{0\leq u}^{T}
	\left(
		\Delta Y_u
			-
		\sum_{i=1}^{n}\frac{\Delta\eta_{u_-}}{\eta_{u_-}}
	\right)
	\\
			=&
Y_0
+ 
\sum_{i=1}^{n}\int_0^s Z_{u_-}^idM^i_u 
- 
\sum_{i=1}^{n}\frac1{2}
\int_0^s(Z_s^i)^2
			d[M^i]_s
-\sum_{0\leq u}^{T}
\left(
\Delta Y_u
-
\sum_{i=1}^{n}Z_{u-}^i\Delta M^i_u
\right)
.
\label{eq_thrm_4}
\end{align}
Since $\eta_T = \mathbb{E}\left[
e^{X_T}
\middle| \mathcal{F}_T
\right] = e^{X_T}$, then for $s \in [0,T]$, equation~\eqref{eq_thrm_4} gives way to the following semimartingale BSDE
\begin{align}
Y_s=&{X_T} - 
Y_0
-
\sum_{i=1}^{n}\int^T_s Z_{u_-}^idM^i_u 
+
\sum_{i=1}^{n}\frac1{2}
\int^T_s(Z_s^i)^2
d[M^i]_s
\\
-&
	\sum_{u\geq s}^{T}
\left(
\Delta Y_u
-
\sum_{i=1}^{n}Z_{u-}^i\Delta M^i_u
\right)
.
\label{eq_thrm_5}
\end{align}
From the condition that $\mathbb{P}^{\star}\in \mathcal{P}_t(1)$, we notice that equation~\eqref{eq:optimal:measure:a} implies that $\eta_t=1$; thus $Y_t=0$.  Moreover, $Y_s=
-\ln\left(
\mathbb{E}_{\mathbb{P}}\left[
e^{X_T}
\middle| \mathcal{F}_s
\right]
\right)
$, hence equation~\eqref{eq_thrm_5} may be rewritten as 
\begin{align}
\nonumber
-\ln\left(
\mathbb{E}_{\mathbb{P}}\left[
e^{X_T}
\middle| \mathcal{F}_s
\right]
\right) -{X_T} = &
-
\sum_{i=1}^{n}\int^T_s Z_{u_-}^idM^i_u 
+
\sum_{i=1}^{n}\frac1{2}
\int^T_s(Z_s^i)^2
d[M^i]_s
\\
-&
\sum_{u\geq s}^{T}
\left(
\Delta Y_u
-
\sum_{i=1}^{n}Z_{u-}^i\Delta M^i_u
\right)
.
\label{eq_thrm_6}
\end{align} 
Taking the exponential of both sides of equation~\eqref{eq_thrm_6} and noting that $\frac{e^{X_T}}{\mathbb{E}_{\mathbb{P}}\left[
	e^{X_T}
	\middle| \mathcal{F}_s
	\right]}$ is the right-hand side of equation~\eqref{eq:optimal:measure:a}, 
yields equations~\eqref{thrm_main_defining_system} and~\eqref{thrm_main_defining_equation}. 
\end{proof}

The connection between entropic measure transform, aggregate relative entropy, affine processes, and risk measures will be discussed in the following section.  
\section{Connections to Risk Measures and Affine Processes}\label{s_gen_connections}
Affine processes, present a large and tractable class of stochastic processes with many desirable properties.  For example, in \cite{cuchiero2011affine} the moment generating function of affine processes were characterized in terms of generalized Riccati equations and subsequently in \cite{2018arXiv180107796G} highly accurate approximate filtering methodologies were developed.  In the next section, the EMT is used to obtain a convenient closed form expression for the conditional expectation of an affine processes, which admits an affine representation under an auxiliary equivalent measure to $\mathbb{P}$.  Connections are subsequently made to the entropic risk measure.

\subsection{Characterization of Cumulants for Affine Processes}\label{ss_affine}
Suppose that $X_t$ is a time-homogeneous Markov process for which there exists $\mathbb{C}$ and $\mathbb{C}^n$-valued functional $p$ and $q$, respectively, such that the characteristic function of $X_t$ may be written as
\begin{equation}
\mathbb{E}_{\mathbb{P}}\left[
	e^{\langle u,X_T\rangle}
\middle| X_t = x\right]
	=
\exp\left(
p(t,u) + \langle q(t,u),X_t\rangle
\right)
\label{eq_affine_property}
.
\end{equation}

Following \cite{KellerEMAFF}, we make the following assumption.  
\begin{assumption}\label{ass_exun}
For every $u,x \in \mathbb{R}^n$ $\mathbb{E}_{\mathbb{P}}^x\left[e^{\langle u,X_t\rangle }\right]<\infty$.
\end{assumption}
Under Assumption~\ref{ass_exun}, it was seen in \cite{KellerEMAFF} show that the following generalized Riccati equations have a unique (real) minimal solution
\begin{align}
\frac{\partial p}{\partial t}(t,y) &= F(q(t,u)) ;  &p(0,y)=0\\
\frac{\partial q}{\partial t}(t,y) & = R(q(t,u));  &q(0,y)=y
\label{eq_affine_gen_riccatti}
,
\end{align}
for of all $t\in [0,T]$, where the functions $F$ and $R$ are defined in \cite{KellerEMAFF}.  Moreover, the unique minimal solution $(p,q)$ was shown to characterize the moments of $X_t$, through
$$
\mathbb{E}_{\mathbb{P}}\left[
e^{\langle u,X_T\rangle}
\middle| X_t = x\right]
=
\exp\left(
p(t,u) + \langle q(t,u),X_t\rangle
\right)
.
$$

\begin{proposition}\label{thrm_affine_characterization}
	Let $\mathbb{Q}$ be a measure equivalent to $\mathbb{P}$, whose EMT is $\mathbb{P}$, that is
	\begin{equation}
	\mathbb{Q}^{\star}=\mathbb{P}
\label{eq_thrm_affine_characterization}
,
	\end{equation}
under which $X_t$ is an $\mathbb{R}$-valued affine process, satisfies Assumption~\ref{ass_exun}.  Then the expectation of $X_t$ can be characterized by
	\begin{equation}
	\mathbb{E}_{\mathbb{P}}\left[%
	X_T%
	\middle |\mathcal{F}_t\right]
		=
	p(T-t,1) + \langle q(T-t,1) , X_t\rangle + 	H_{t,T}(\mathbb{P}|\mathbb{Q})
	\label{eq_thrm_affine_characterization_result}
	. 
	\end{equation}
\end{proposition}
\begin{proof}
Let $(p,q)$ be the minimal solution to the Riccati system of equation \eqref{eq_affine_gen_riccatti}, under $\mathbb{Q}$.  Hence, \citep[Corollary 2.16]{KellerEMAFF} implies that for every $x \in \mathbb{R}^{n}$, $\theta \in \mathbb{R}$, $u \in S^{n}$, and $t \in [0,T]$ it follows that
\begin{equation}
\mathbb{E}_{\mathbb{P}}\left[
e^{\theta \langle u,X_T\rangle } \middle| \mathcal{F}_t
\right]
	=
\exp\left(
	p(T-t,\theta\cdot u) + \langle q(T-t,\theta \cdot u) , X_t\rangle 
\right)
\label{eq_thrm_affine_characterization_1}
.
\end{equation}
For $\theta \neq 0$, combining equations~\eqref{eq_thrm_affine_characterization_1} and~\eqref{eq:duality:relation:a} yields
\begin{align}
\nonumber
p(T-t,u) + \langle q(T-t,u) , X_t\rangle  = &
\frac{p(T-t,\theta\cdot u) + \langle q(T-t,\theta \cdot u) , X_t\rangle }{\theta} \\
	 = &
\mathbb{E}_{\mathbb{P}^{\star}}\left[\langle u,X_T\rangle\middle |\mathcal{F}_t\right]-\frac1{\theta}
	H_{t,T}(\mathbb{P}^{\star}|\mathbb{P})
\label{eq_thrm_affine_characterization_2}
.
\end{align}
Setting $u=1=\theta$, combining equations~\eqref{eq_thrm_affine_characterization_2}, equation~\ref{eq_thrm_affine_characterization}, and equation~\eqref{eq:duality:relation:a} yields equation~\eqref{eq_thrm_affine_characterization_result}.  
\end{proof}
\subsection{Relation To Conditional Entropic Risk}\label{ss_ce_risk}

The conditional entropic risk-measure with risk-aversion level $\theta>0$, denoted by $\rho_{t,\theta}^{ent}$, is a well-studied dynamic convex risk-measure.  Defined on suitable random variable (see \cite{DCRMs} for details) by
$$
\rho_{t,\theta}^{ent}(X)\triangleq \frac1{\theta}\ln\left(\mathbb{E}_{\mathbb{P}}\left[e^{-\theta X}\mid \mathcal{F}_t\right]\right).
$$
Theorem~\ref{thrm_main_defining_equation}, Proposition~\ref{proposition:duality:relation:a}, and \citep[Theorem 3.8]{JacodShiryaevLimitTheorems}, give a way to explicitly compute the conditional entropic-risk of a $\mathcal{F}_T$-measurable random-variable via the formula
\begin{align}\label{eq_CER}
\rho_{t,\theta_t}^{ent}(X_T) = &\mathbb{E}_{\mathbb{P}^{\star}}\left[
X_T
\middle| \mathcal{F}_t \right]
-
\frac{H_{t,T}\left(\mathbb{P}^{\star}|\mathbb{P}\right)}{\theta_t}
,
\end{align}
where the risk-aversion level $\theta_t$ is $X_T$-compatible.  The predictability of the risk-aversion level is interpreted as the ability to adjust risk-aversion as new information arrives.  Equation~\eqref{eq_CER} can be interpreted as saying, that the conditional entropic risk with respect to $X_t$ is equivalent to a correction of its conditional expectation, under $\mathbb{P}^{\star}$.  The correction for the undertaken risk-correction is realized through the term $H_{t,T}(\mathbb{P}^{\star}|\mathbb{P})$.  

The financial interpretations of the entropic measure transform are explored in the next section.  The examples will focus on the term-structure of interest rates, defaultable bonds, futures prices, and forward prices.  

\section{Pricing via the Entropic Measure Transform}\label{s_EMT_if}
We begin by exploring the connections and application of Theorem~\eqref{thrm_main} in mathematical finance. Specifically, we consider applications to bond pricing, futures prices, and forward prices.  We first explore the connection of the entropic measure transform to bond pricing.  

\subsection{Short-Rate Models for Defaultable Bond Pricing}\label{default}
Let $T> 0$, be the investment horizon and $\mathbb{P}$ is a martingale measure for the for the short-rate, using the money market account as num\'{e}raire.  The short-rate will be modeled as $r(X_t)$, where $r$ is a twice-continuously differentiable function from $\mathbb{R}^n$ to $(0,\infty)$ and , $\mathcal{F}_s$-adapted factor process following
\begin{equation}\label{eq:SDE:forward:a}
dX_s=f(s,X_{s-})ds+g(s,X_{s-})dW^\mathbb{P}_s+zN^\mathbb{P}(ds,dz),
\end{equation}
where $W^\mathbb{P}$ is an $n$-dimensional $(\mathcal{F},\mathbb{P})$-Brownian motion and where the random measure $N^\mathbb{P}$ is an $\mathbf{R}^n$-valued random measure with compensator
\begin{equation*}
\eta(ds,dz)=v(dz)\lambda(X_{s-})ds,
\end{equation*}
where $v(\cdot)$ is a measure on $\mathbf{R}^n$, $\lambda(\cdot)$ is a function to be specified from $\mathbf{R}^n$ to $\mathbf{R}$. The  compensated random measure is
\begin{equation}
\tilde N^\mathbb{P}(dt,dz)=N^\mathbb{P}(dt,dz)-v(dz)\lambda(X_{t-})dt.
\label{eq_random_measure_descr}
.
\end{equation}

The price of a default-free zero-coupon bond at time $t\in 
[0,T]$ is then given by
\begin{equation} \label{eq:bondprice:a}
P(t,T)=\mathbb{E}_{\mathbb{P}}\left[e^{-\int_{t}^{T}r(X_s)ds} \middle|\mathcal{F}_t\right]
; \quad T\geq t\geq 0
.
\end{equation}

We consider a defaultable zero coupon bond with the promised payoff of $\$1$ at maturity, and denote the price at time $t\in[0,T]$ by $D(t,T)$. Unlike default-free bonds, the issuer of defaultable bonds, such as corporate bonds, may default before the maturity  in which case the bondholders will not receive the promised payment in full but a recovery payment. There are different recovery schemes  if default occurs before the bond's maturity according to the timing and the amount of recovery payment (see \citet[Section 1.1.1]{bielecki2002credit} and \citet{altman2004default}). For instance, if a fixed fraction of the bond's face value is paid to the bondholder at maturity $T$ in case of default, then the bond has the random payoff at maturity
$$C_T=\mathbf{1}_{\{\tau>T\}}+\eta \mathbf{1}_{\{\tau\leq T\}}$$
where $\tau$ is the default time. If a fixed fraction of the pre-default market value of the bond value is paid  at time of default, then the equivalent random payoff of the bond is
$$C_T=\mathbf{1}_{\{\tau>T\}}+\eta P(\tau-,T)e^{\int_\tau^Tr_vdv} \mathbf{1}_{\{\tau\leq T\}}.$$
The time of default $\tau$ is also modelled differently. Under the structural credit risk models originating with \citet{merton1974pricing} the default of corporate bonds occurs when the value of the firm reaches a certain lower threshold. Reduced form credit risk models, such as \citet{duffie1999modeling}, assume that default is driven by  an exogenous default process.

Since the recovery scheme is not our main concern in this paper, we will in general represent the equivalent  payoff of defaultable bonds with a random payoff $C_T$, and assume that the price is given by (see \citet{duffie1999modeling})
\begin{equation}
D(t,T)=\mathbb{E}_{\mathbb{P}}\left[
e^{-\int_{t}^{T}r(X_v)dv}\cdot C_T \middle|\mathcal{F}_t
\right], 
\label{eq_defaultable_priceosku}
\end{equation}
where $C_T$ is an $\mathcal{F}_T$-measurable random variable valued in $[0,1]$.

In the extreme situation $C_T=0$ of a complete default, in which the bondholders receive no recovery payment in the event of default,  the bonds become worthless. In this paper, we exclude the occurrence of complete default by assuming $\mathbb{P}(C_T=0)=0$. We will later explain why we have to make this technical assumption. The other extreme case of default-free bonds is included in our model if we assume $\mathbb{P}(C_T=1)=1$.  The default-free case is captured by equation~\eqref{eq_defaultable_priceosku}, if $\mathbb{P}\left(C_T=1\right)=1$.  

Theorem~\ref{thrm_main} and equation~\eqref{eq:SDE:forward:a}, directly implies that the defaultable bond price $D(t,T)$ is characterized by the following EMT problem.  
\begin{corollary}\label{cor_bond_pricing}
	Under the modeling scheme described by equations~\eqref{eq:SDE:forward:a} and \eqref{eq:bondprice:a} the price of a zero-coupon bond is
	\begin{align}
	\label{cor_bond_pricing_eq_descriptive}
	D(t,T) 
	= &
	\exp\left(
	\mathbb{E}_{\mathbb{P}^{\star}}\left[
	\int_t^T r(X_u)du - \ln(C_T)\middle | \mathcal{F}_t
	\right]
	\right)\kappa_t,\\
	\nonumber
	\kappa_t\triangleq &
	\exp\left(
	\mathbb{E}_{\mathbb{P}}\left[
	\int_t^T\frac{1}{2}Z_vZ_v^\prime
	dv-\int_t^T Z_vdW^\mathbb{P}_v
	\middle| \mathcal{F}_t\right]
	\right),%
	\end{align}
	where $Z_t$ is defined through the decoupled FBSDE
	\begin{align}
	&X_s=X_t+\int_t^s f(v,X_{v-})ds+\int_t^s
	g(v,X_{v-})dW_v^\mathbb{P}+\int_t^s\int_{\mathbf{R}^n}zN^\mathbb{P}(dv,dz),\label{FBSDE:jumps:forward:a}	\\
		\label{FBSDE:jumps:backward:a}
	&Y_s=-\ln(C_T) + \int_s^T\Big\{r(X_{v-})-\big[\int_{\mathbf{R}^n}(e^{G(v,z)}-1)v(dz)\big]\lambda(X_{v-})-
	\frac{1}{2}Z_vZ_v^\prime\Big\}dv   \\
		\nonumber
	&\quad \quad +\int_s^TZ_vdW_v^\mathbb{P}+\int_s^T\int_{\mathbf{R}^n}G(v,z)N^\mathbb{P}(dv,dz),	
\\
	\nonumber
	& Y_s\triangleq \ln \left\{\mathbb{E}_{\mathbb{P}}[e^{-\int_{t}^{T}r(X_u)du}\cdot C_T |\mathcal{F}_s]\right\}.
	\end{align}
	Here, the entropic measure transform, $\mathbb{P}^{\star}$ of $\mathbb{P}$, is given by
	\begin{align}
	\left.\frac{d \mathbb{Q}^\star}{d\mathbb{P}}\right|_{\mathcal{F}_T}&=\exp\left\{-	 
	\int_t^T\frac{1}{2}Z_vZ_v^\prime dv+\int_t^T
	Z_vdW^\mathbb{P}_v-\int_t^T\int_{\mathbf{R}^n}\lambda(X_{v-})(e^{G(v,z)}-1)v(dz)dv\right.\nonumber\\
	&\left.~~~+ \int_t^T\int_{\mathbf{R}^n}G(v,z)N^\mathbb{P}(dv,dz)\right\}.
	\label{eq:EMT:jumps:optimal:measure:a}	
	\end{align}

\end{corollary}
\begin{proof}
Combining the backwards equation~\eqref{thrm_main_defining_system}, equation~\eqref{eq_random_measure_descr} and the SDE \eqref{eq:SDE:forward:a} forms the decoupled FBSDE of equation~\eqref{FBSDE:jumps:forward:a}.  
\end{proof}
\begin{remark}
	Suppose a financial agent pays $c$ to buy one unit of the bond at time t, and receives a payoff of $C_T$  at maturity $T$. The internal logarithmic return on the investment over the time period $[t,T]$ is
	$$\gamma=\ln \frac{C_T}{c}.$$ 
	The excess return over the risk-free rate, $\tilde\gamma$,  is given by
	$$\tilde\gamma=\gamma-\int_{t}^{T}r(X_v)dv,$$
	which measures the investment performance.
	Note that the entropic measure transform of $\mathbb{P}$, in the defaultable bond price setting, is equivalent to 
	\begin{align}
	\notag \ln\frac{D(t,T)}{c}&=-\inf_{\mathbb{Q}\in \mathcal{P}_t(1)}\big\{\mathbb{E}_{\mathbb{Q}}[-\tilde\gamma|\mathcal{F}_t]+H_{t,T}(\mathbb{P}^{\star}|\mathbb{P})\big\}\\
	&=\sup_{\mathbb{Q}\in
		\mathcal{P}_t(1)}\big\{\mathbb{E}_{\mathbb{Q}}[\tilde\gamma|\mathcal{F}_t]-H_{t,T}(\mathbb{P}^{\star}|\mathbb{P})\big\}.\label{eq:EMT:defautable:interpretation}
	\end{align}
	The aggregate relative entropy $H_{t,T}(\mathbb{P}^{\star}|\mathbb{P})$ in equation
	(\ref{eq:EMT:defautable:interpretation}) can be interpreted as penalty for removing financial risk
	composed of market risk (volatility risk) and credit risk in the framework of our model. The right-hand
	side of equation (\ref{eq:EMT:defautable:interpretation})  maximizes the  excess  (risk-adjusted) return
	on the investment, which is equal to the equivalent instantaneous return given by left-hand side of
	equation (\ref{eq:EMT:defautable:interpretation}).
\end{remark}

We will discuss the explicit solution to the FBSDE \eqref{FBSDE:jumps:forward:a}-\eqref{FBSDE:jumps:backward:a} in the case of ATSMs and QTSMs, respectively.  The possibility of default leads to solutions with an extra component compared to those considered by \citet{hyndman2009forward} and \citet{Zhou2010}.

\subsubsection{Non-Defaultable Bond Case with ATSM}
In the framework of ATSMs with jumps, we make the following specifications on the coefficients of FBSDE  \eqref{FBSDE:jumps:forward:a}-\eqref{FBSDE:jumps:backward:a} as follows
\begin{enumerate}[(i)]
	\item $f(s,x)=Ax+B$
	\item $g(s,x)=S\text{diag}\sqrt{\alpha_i+\beta_i x}$
	\item $r(x)=R^\prime x+k$
	\item $\lambda(x)=L^\prime x+l$
\end{enumerate}
where $A$ is an $(n\times n)$-matrix of scalars, $B$, $R$ and $L$ are $(n\times 1)$-vectors, for each
$i\in\{1,\ldots,n\}$ the $\alpha_i$ are scalars,  for each $i\in\{1,\ldots,n\}$ the
$\beta_i=(\beta_{i1},\ldots,\beta_{in})$ are $(1\times n)$-vectors, $S$ is a non-singular $(n\times
n)$-matrix, $k$ and $l$ are  scalars. Then FBSDE  (\ref{FBSDE:jumps:forward:a})-(\ref{FBSDE:jumps:backward:a}) becomes 
\begin{align}
X_s&=X_t+\int_t^s\big[AX_{v-}+B\big]dv+\int_t^s S\text{diag}\sqrt{\alpha_i+\beta_i X_{v-}}dW^{\mathbb{P}}_v
+\int_t^s\int_{\mathbf{R}^n}z N^{\mathbb{P}}(dv,dz),\label{FBSDE:jumps:ATSM:forward:a}\\
Y_s&=\int_s^T\Big\{\big[\big(\int_{\mathbf{R}^n}(e^{G(v,z)}-1)v(dz)\big) L^\prime-R^\prime
\big]X_{v-}-\big(\int_{\mathbf{R}^n}(e^{G(v,z)}-1)v(dz)\big) l+k-\frac{1}{2}Z_vZ^\prime_v\Big\}dv \nonumber\\ 
&\quad +\int_s^TZ_vdW^{\mathbb{P}}_v +\int_s^T\int_{\mathbf{R}^n}G(v,z)N^{\mathbb{P}}(dv,dz).\label{FBSDE:jumps:ATSM:backward:a} 
\end{align}
We will give the explicit solution to FBSDE
(\ref{FBSDE:jumps:ATSM:forward:a})-(\ref{FBSDE:jumps:ATSM:backward:a}) by applying a similar technique to \citet{hyndman2009forward} which extends the approach for linear FBSDEs from \citet{ma1999forward}. In the statement of the following proposition, as in \citet{hyndman2009forward}, we shall adopt the notation of \citet{landen2002term} to write
$$S\text{diag}(\alpha_i+\beta_ix)S^\prime=k_0+\sum_{j=1}^{n}k_jx_j$$
for symmetric $(n\times n)$ matrices $k_j$, where $x_j$ is the $j$th element of a vector $x\in D$. Define the $(n^2\times n)$ matrix $K$ and, given a $(1\times n)$ row vector $\mathbf{\underline{y}}$, the $n\times n^2$ matrix $\beta(\mathbf{\underline{y}})$ by
$$K=
\begin{bmatrix}
k_1\\
k_2\\
\vdots\\
k_n
\end{bmatrix}
\quad \text{and}
\quad \beta(\underline{y})=
\begin{bmatrix}
\underline{y}&0_{1\times n}&\cdots&0_{1\times n}\\
0_{1\times n}&\underline{y}&&\\
\vdots&&\ddots&\vdots\\
0_{1\times n}&\cdots&&\underline{y}
\end{bmatrix}
$$
respectively.

\begin{theorem}\label{theorem:BSDE:jump:solution}
	If the Riccati equation
	\begin{align}
	&\dot{U}_s+U_sA+\frac{1}{2}U_sK^\prime[\beta(U_s)]+\Big[\int_{\mathbf{R}^n}(e^{U_sz}-1)v(dz)\Big]L^\prime-
	R^\prime=0, \quad t\in[0,T] 	\label{eq:Riccati:ATSM:jump}\\
	&U_T=0,	\label{eq:Riccati:ARSM:jump:ter:condtion}
	\end{align}
	admits a unique bounded solution $U(\cdot)$ over the interval $[0,T]$, then the FBSDE
	(\ref{FBSDE:jumps:ATSM:forward:a})-(\ref{FBSDE:jumps:ATSM:backward:a}) admits a unique solution and
	$(Y,Z,G)$ has explicit expression in terms of $X$ as follows
	\begin{align} 
	&Y_s= -(U_sX_s+p_s), \label{eq:theorem:ATSM:jump:solution:Y}\\
	&Z_s=U_sS\text{diag}(\sqrt{\alpha_i+\beta_i X_{s-}}), %
	\label{eq:theorem:ATSM:jump:solution:Z}\\
	&G(s,z)=U_sz, \label{eq:theorem:ATSM:jump:solution:G}
	\end{align}
	where $p_s$ is given by
	\begin{equation}\label{eq:theorem:ATSM:jump:solution:p}
	p_s=-\int_s^T\left(k-l\int_{\mathbf{R}^n}(e^{U_vz}-1)v(dz)-\frac{1}{2}U_vk_0U_v^\prime-U_vB\right) dv
	\end{equation}
\end{theorem}
\begin{proof}
	We first prove the decoupled FBSDE (\ref{FBSDE:jumps:ATSM:forward:a})-(\ref{FBSDE:jumps:ATSM:backward:a}) admits a unique solution $(X,Y,Z,G)$. 
	The SDE (\ref{FBSDE:jumps:ATSM:forward:a}) admits a unique solution. As $X_s$ is known, we consider the
	single BSDE (\ref{FBSDE:jumps:ATSM:backward:a}).
	If we let 
	$$	\tilde Y_s=e^{-Y_s}, \quad	\tilde Z_s=-\tilde Y_s\cdot Z_s, \quad \tilde G(z,s)=-\tilde Y_s(1-e^G(s,z)),$$
	BSDE (\ref{FBSDE:jumps:ATSM:backward:a}) becomes
	\begin{align}
	\tilde Y_s&=1+\int_s^T\Big\{  \big[\big(\int_{\mathbf{R}^n}(e^{G(v,z)}-1)v(dz)\big) L^\prime-R^\prime
	\big]X_{v-}-\big(\int_{\mathbf{R}^n}(e^{G(v,z)}-1)v(dz)\big) l+k  \Big\}\tilde Y_vdv \nonumber\\
	&\quad +\int_s^T\tilde Z_vdW^{\mathbb{P}}_v +\int_s^T\int_{\mathbf{R}^n}\tilde G(v,z)\tilde
	N^{\mathbb{P}}(dv,dz). \label{eq:theorem:ATSM:jump:proof:a}
	\end{align}
	By \citet[Theorem 3.1.1]{delong2013backward}, we know that the BSDE (\ref{eq:theorem:ATSM:jump:proof:a}) admits a unique solution
	$(\tilde Y,\tilde Z,\tilde G)$.  Therefore, the FBSDE
	(\ref{FBSDE:jumps:ATSM:forward:a})-(\ref{FBSDE:jumps:ATSM:backward:a})  admits a unique solution $(X,Y,Z,G)$.
	
	To prove the explicit expression of $(Y,Z,G)$, we need to show that $(Y,Z,G)$ given by equations
	(\ref{eq:theorem:ATSM:jump:solution:Y})-(\ref{eq:theorem:ATSM:jump:solution:G}) satisfies the BSDE
	(\ref{FBSDE:jumps:ATSM:backward:a}). 
	Apply It\^{o}'s formula to the function $\phi(s,x)=-(U_sx+p_s)$,  where $U_s$ is the solution to the
	Riccati equation (\ref{eq:Riccati:ATSM:jump}) and $p_s$ satisfies equation
	(\ref{eq:theorem:ATSM:jump:solution:p}). Let $Y_s=\phi(s,X_s)$, where $X_s$ is given by equation
	(\ref{FBSDE:jumps:ATSM:forward:a}), then we have
	\begin{align}
	&\quad  Y_T-Y_s \nonumber\\
	&=  -\int_s^T\left(\dot{U}_vX_{v-}+U_v(AX_{v-}+B)+k_0U_v^\prime+K^\prime[\beta(U_v)]^\prime
	X_{v-}\right)dv \nonumber\\
	&\quad -\int_s^TU_v
	S\text{diag}(\sqrt{\alpha_i+\beta_iX_v})dW^{\mathbb{P}}_v-\int_s^T\Big(k-l\int_{\mathbf{R}^n}(e^{U_vz}-1)v(dz)-\frac{1}{2}U_vk_0U_v^\prime
	-U_vB\Big) dv\nonumber\\ 
	&\quad-\int_s^T\int_{\mathbf{R}^n}U_vz\tilde N^{\mathbb{P}}(dv,dz)\nonumber\\
	&=-\int_s^T\bigg\{\Big(\dot{U}_v+U_vA+\frac{1}{2}U_vK^\prime[\beta(U_v)]+\Big[\int_{\mathbf{R}^n}(e^{U_vz}-1)v(dz)\Big]L^\prime-
	R^\prime\Big)X_v+\Big[R^\prime X_{v-}\nonumber\\
	&\quad +k+\frac{1}{2}\Big(U_vK^\prime [\beta(U_v)]^\prime
	X_{v-}+U_vk_0U^\prime_v\Big)\Big]\bigg\}dv+\int_s^T\Big(\Big[\int_{\mathbf{R}^n}(e^{U_vz}-1)v(dz)\Big](L^\prime
	X_{v-}+l)\Big)dv\nonumber\\
	&\quad -\Big\{U_vS\text{diag}(\sqrt{\alpha_i+\beta_i
		X_{v-}})\Big\}dW_v^{\mathbb{P}}-\int_s^T\int_{\mathbf{R}^n}U_vz\tilde N^{\mathbb{P}}(dv,dz) \label{eq:theorem:ATSM:jump:proof:b} 
	\end{align}
	Substituting equations (\ref{eq:theorem:ATSM:jump:solution:Y})-(\ref{eq:theorem:ATSM:jump:solution:G})
	into equation (\ref{eq:theorem:ATSM:jump:proof:b}) we have
	\begin{align*}
	Y_s&=Y_T+\int_s^T(R^\prime X_{v-}+k+\frac{1}{2}Z_vZ^\prime_v)dv-\int_s^T\int_{\mathbf{R}^n}(L^\prime X_{v-}+l)\big(e^{G(v,z)}-1\big)v(dz)dv\\
	&\quad +\int_s^TZ_vdW^{\mathbb{P}}_v +\int_s^T\int_{\mathbf{R}^n}G(s,z)N^{\mathbb{P}}(ds,dz)
	\end{align*}
	By the boundary condition of \eqref{eq:Riccati:ARSM:jump:ter:condtion} and \eqref{eq:theorem:ATSM:jump:solution:p} we have
	\begin{equation*}
	Y_T=-(U_TX_T+p_T)=0.
	\end{equation*}
	Therefore,
	\begin{align*}
	Y_s&=\int_s^T(R^\prime X_{v-}+k+\frac{1}{2}Z_vZ^\prime_v)dv-\int_s^T\int_{\mathbf{R}^n}(L^\prime X_{v-}+l)\big(e^{G(v,z)}-1\big)v(dz)dv\\
	&\quad +\int_s^TZ_vdW^{\mathbb{P}}_v +\int_s^T\int_{\mathbf{R}^n}G(v,z)N^{\mathbb{P}}(dv,dz)
	\end{align*}
	Hence $(Y,Z,G)$ given by equations
	(\ref{eq:theorem:ATSM:jump:solution:Y})-(\ref{eq:theorem:ATSM:jump:solution:G}) satisfy BSDE
	(\ref{FBSDE:jumps:ATSM:backward:a}).
\end{proof}

\begin{remark}
	The complete discussion on the Riccati equation of the form as in \eqref{eq:Riccati:ATSM:jump} can be found in \citet[Section 6]{duffie2003affine}.
\end{remark}

\subsubsection{Defaultable Case with ATSM and No Jumps}

Under the framework of ATSMs, with no jumps the FBSDE \eqref{FBSDE:jumps:forward:a}-\eqref{FBSDE:jumps:backward:a} becomes 
\begin{align}
&X_s=X_t+\int_t^s \left(AX_v+B\right)dv+\int_t^sS\text{diag}\sqrt{\alpha_i+\beta_i X_v}dW^{\mathbb{P}}_v \label{FBSDE:defautable:ATSM:forward:a}\\
&Y_s=-\ln C_T+\int_s^T(R^\prime X_v+k-\frac{1}{2}Z_vZ_v^\prime)dv+\int_s^TZ_vdW^{\mathbb{P}}_v \label{FBSDE:defautable:ATSM:backward:a}
\end{align}
The following result can be seen as a generalization of \citet[Theorem 3.2]{hyndman2009forward} by incorporating a random terminal condition representing the recovery amount in the case of default.
\begin{theorem} \label{theorem:FBSDE:defautable:solution}
	If the Riccati equation 
	\begin{align}
	&\dot{U}_s+U_sA+\frac{1}{2}U_sK^\prime[\beta(U_s)]- R^\prime=0, \quad s\in[0,T] \label{eq:Riccati:defautable}\\
	&U_T=0 \label{eq:Riccati:defaultable:boundary}
	\end{align}
	admits a unique bounded solution $U(\cdot)\in \mathbf{R}^n$ over the interval $[0,T]$, then  FBSDE
	\eqref{FBSDE:defautable:ATSM:forward:a}-\eqref{FBSDE:defautable:ATSM:backward:a} admits a unique solution
	and the solution $(Y,Z)$ has explicit expression in terms of $X$
	\begin{align}
	Y_s&= -(U_sX_s+p_s),\quad \text{and} \label{eq:FBSDE:defautable:rep:Y}	\\
	Z_s&=U_sS\text{diag}(\sqrt{\alpha_i+\beta_i X_s})+z_s, \label{eq:FBSDE:defautable:rep:Z}	
	\end{align}
	where $(p_s, z_s)$ solves the following BSDE
	\begin{equation}\label{eq:FBSDE:defautable:p}
	p_s=-\ln C_T-\int_s^T\left(k-\frac{1}{2}U_vk_0U_v^\prime-U_vB+\frac{1}{2}z_v z_v^\prime\right) dv-\int_s^Tz_vdW^{\mathbb{P}}_v.  
	\end{equation}
\end{theorem}
\begin{proof}
	We first prove the decoupled FBSDE \eqref{FBSDE:defautable:ATSM:forward:a}-\eqref{FBSDE:defautable:ATSM:backward:a} admits a unique solution $(X,Y,Z)$. 
	Under our assumptions the SDE~\eqref{FBSDE:defautable:ATSM:forward:a} admits a unique solution. Given $X_s$, we consider the  BSDE \eqref{FBSDE:defautable:ATSM:backward:a}.
	If we let 
	$$	\tilde Y_s=e^{-Y_s}, \quad	\tilde Z_s=-\tilde Y_s\cdot Z_s,$$
	the BSDE \eqref{FBSDE:defautable:ATSM:backward:a} becomes
	\begin{equation}
	\tilde Y_t=C_T+\int_t^T\big[  R^\prime X_s+k  \big]\tilde Y_sds +\int_t^T\tilde Z_sdW^{\mathbb{P}}_s.
	\label{c010708} 
	\end{equation}
	Clearly the BSDE \eqref{FBSDE:defautable:ATSM:backward:a} admits a unique solution
	$(\tilde Y,\tilde Z)$ so the FBSDE
	\eqref{FBSDE:defautable:ATSM:forward:a}-\eqref{FBSDE:defautable:ATSM:backward:a}  admits a unique solution $(X,Y,Z)$.
	Using the same technique, we can also prove BSDE \eqref{eq:FBSDE:defautable:p} admits a unique solution $(p,z)$.
	
	To prove the explicit representation of $(Y,Z)$, we need to show $(Y,Z)$ given by equations
	(\ref{eq:FBSDE:defautable:rep:Y})-(\ref{eq:FBSDE:defautable:rep:Z})	 satisfies the BSDE \eqref{FBSDE:defautable:ATSM:backward:a}. 
	Apply It\^{o}'s formula to the function $\phi(s,x,p)=-(U_s x+p)$  where $U_s$ is the solution to \eqref{eq:Riccati:defautable}. Let $Y_s=\phi(s,X_s, p_s)$ where $X_s$ is given by \eqref{FBSDE:defautable:ATSM:forward:a} and $p_s$ satisfies \eqref{eq:FBSDE:defautable:p}.  Then we have
	\begin{align}
	\notag  dY_s &=  -\left(\dot{U}_sX_s+U_s(AX_s+B+k_0U_s^\prime+K^\prime[\beta(U_s)]^\prime X_s+S\text{diag}(\sqrt{\alpha_i+\beta_iX_s})z_s^\prime\right)ds\\
	\notag  &-U_s S\text{diag}(\sqrt{\alpha_i+\beta_iX_s})dW^\mathbb{P}_s-\left(k-\frac{1}{2}U_sk_0U_s^\prime-U_sB+\frac{1}{2}z_s z_s^\prime\right) ds-z_sdW^{\mathbb{P}}_s \\
	\notag  &=-\bigg\{\Big(\dot{U}_s+U_sA+\frac{1}{2}U_sK^\prime[\beta(U_s)]- R^\prime\Big)X_s+\Big[R^\prime X_s+k+\frac{1}{2}\Big(U_sK^\prime [\beta(U_s)]^\prime X_s+U_sk_0U^\prime_s\\
	&+2U_sS\text{diag}(\sqrt{\alpha_i+\beta_i X_s})z^\prime_s+z_s z_s^\prime\Big)\Big]\bigg\}ds-\Big\{U_sS\text{diag}(\sqrt{\alpha_i+\beta_i X_s})+z_s\Big\}dW_s^{\mathbb{P}}. \label{eq:theorem:FBSDE:defautable:proof:a}
	\end{align}
	Substituting equations \eqref{eq:Riccati:defautable} and \eqref{eq:FBSDE:defautable:rep:Z} into
	equation \eqref{eq:theorem:FBSDE:defautable:proof:a} we have
	\begin{equation*}
	dY_s=-(R^\prime X_s+k+\frac{1}{2}Z_sZ_s^\prime)ds-Z_sdW^{\mathbb{P}}_s
	\end{equation*}
	Thus $(Y_s,Z_s)$ defined by equations \eqref{eq:FBSDE:defautable:rep:Y}-\eqref{eq:FBSDE:defautable:rep:Z}	 satisfies
	\begin{equation*}
	Y_s=Y_T+\int_s^T(R^\prime X_v+k+\frac{1}{2}Z_vZ_v^\prime)dv+\int_s^TZ_vdW^{\mathbb{P}}_v
	\end{equation*}
	By the boundary conditions in equations \eqref{eq:Riccati:defaultable:boundary} and \eqref{eq:FBSDE:defautable:p} we have
	\begin{equation*}
	Y_T=-\ln C_T
	\end{equation*}
	Therefore,
	\begin{equation*}
	Y_s=-\ln C_T+\int_s^T(R^\prime X_v+k-\frac{1}{2}Z_vZ_v^\prime)dv+\int_s^TZ_vdW^{\mathbb{P}}_v.
	\end{equation*}
\end{proof}

\begin{remark}
	The existence and uniqueness of the solution to the Riccati equation \eqref{eq:Riccati:defautable} is shown  in \citet[Section 6]{duffie2003affine} where  a class of generalized Riccati equations has been considered.
\end{remark}
Note that the representation of $(Y,Z)$ of the FBSDE
(\ref{eq:FBSDE:defautable:rep:Y})-(\ref{eq:FBSDE:defautable:rep:Z}) is not completely explicit, since the term $z_t$ is to be determined by the quadratic BSDE \eqref{eq:FBSDE:defautable:p}.  Fortunately we can convert the quadratic BSDE \eqref{eq:FBSDE:defautable:p} into a linear BSDE by letting
$$	\tilde p_t=e^{-p_t}, \quad	\tilde z_t=\tilde p_t\cdot z_t,$$
then the BSDE \eqref{eq:FBSDE:defautable:p} becomes
\begin{equation}\label{eq:FBSDE:defautable:numerical}
\tilde p_t=\frac{1}{C_T}+\int_t^T(k-\frac{1}{2}U_sk_0U_s^\prime-U_sB)\tilde p_sds+\int_t^T\tilde
z_sdW^{\mathbb{P}}_s.
\end{equation}
In the excluded case that $P(C_{T}=0)>0$ then (\ref{eq:FBSDE:defautable:numerical}) would be a BSDE with singular terminal condition.

With further specification of $C_T$ through a specific a default mechanism and recovery scheme the linear BSDE \eqref{eq:FBSDE:defautable:numerical} can either be solved analytically or numerically.  There is an extensive literature focused on the numerical solution schemes for BSDEs which we shall not discuss. Nevertheless, Theorem~\ref{theorem:FBSDE:defautable:solution} simplifies the procedure to solve the coupled nonlinear FBSDE \eqref{FBSDE:defautable:ATSM:forward:a}-\eqref{FBSDE:defautable:ATSM:backward:a}  to the solution of the Riccati equation \eqref {eq:Riccati:defautable} and the  linear BSDE \eqref{eq:FBSDE:defautable:numerical}.

\subsubsection{Defaultable Case with QTSM and No Jumps}
In the framework of QTSMs with no-jumps, the FBSDE \eqref{FBSDE:jumps:forward:a}-\eqref{FBSDE:jumps:backward:a} becomes
\begin{align}
&X_s=X_t+\int_t^s\big(AX_v+B\big)dv+\int_t^s \Sigma dW^{\mathbb{P}}_v \label{FBSDE:defautable:QTSM:forward}\\
&Y_s=-\ln C_T+\int_s^T(X_v^\prime QX_v+R^\prime
X_v+k-\frac{1}{2}Z^\prime_vZ_v)dv+\int_s^TZ_vdW^{\mathbb{P}}_v.\label{FBSDE:defautable:QTSM:backward}
\end{align}
As in the case of ATSMs we obtain the partially explicit solutions to the FBSDE
\eqref{FBSDE:defautable:QTSM:forward}-\eqref{FBSDE:defautable:QTSM:backward} stated in the following
theorem. We omit the proof since it is similar to the proof of Theorem~\ref{theorem:FBSDE:defautable:solution}.
\begin{theorem}\label{theorem4.6}
	If the Riccati equations
	\begin{align}
	&\dot{q}_s+q_sA+A^\prime q_s+\frac{1}{2}(q_s^\prime+q_s)\Sigma\Sigma^\prime(q_s^\prime+q_s)-Q=0_{n\times
		n}, \quad s\in[0,T] \label{FBSDE:defautable:QTSM:Riccati:q}\\
	&\dot{u}_s+u_sA+B^\prime(q_s^\prime+q_s)+u_s\Sigma\Sigma^\prime(q_s^\prime+q_s)-R^\prime=0_{1\times n},
	\quad s\in[0,T]\label{FBSDE:defautable:QTSM:Riccati:u}\\
	&q_T=0_{n\times n},\quad u_T=0_{1\times n} \label{FBSDE:defautable:QTSM:Riccati:boundary}
	\end{align}
	admit  unique bounded solutions $q(\cdot)$, $u(\cdot)$ over the interval $[0,T]$, then the FBSDE
	\eqref{FBSDE:defautable:QTSM:forward}-\eqref{FBSDE:defautable:QTSM:backward} admits a unique solution
	and $(Y,Z)$ has explicit expression in terms of $X$ as follows
	\begin{align}
	Y_s&= -(X_s^\prime q_sX_s+u_tX_s+p_s), \label{FBSDE:defautable:QTSM:rep:Y}\\
	Z_s&=\Big(X_s^\prime(q_s+q_s^\prime)+u_s\Big)\Sigma+z_s, \label{FBSDE:defautable:QTSM:rep:Z}
	\end{align}
	where $(p_s, z_s)$ solves the following BSDE
	\begin{equation}
	p_s=-\ln
	C_T-\int_s^T\left(k-u_vB-\frac{1}{2}\text{tr}\big((q_v+q_v^\prime)\Sigma\Sigma^\prime\big)-\frac{1}{2}u_v\Sigma\Sigma^\prime
	u_v^\prime+\frac{1}{2}z_v z_v^\prime\right) dv-\int_s^Tz_vdW^{\mathbb{P}}_v.  \label{FBSDE:defautable:QTSM:p}
	\end{equation} 
\end{theorem}
By the same technique as in  the ATSM case we make the change of  variables 
$$	\tilde p_s=e^{-p_s}, \quad	\tilde z_s=\tilde p_s\cdot z_s,$$
so that the BSDE \eqref{FBSDE:defautable:QTSM:p} to obtain the linear BSDE
\begin{equation}
\tilde p_s=\frac{1}{C_T}+\int_s^T(k-u_vB-\frac{1}{2}\text{tr}\big((q_v+q_v^\prime)\Sigma\Sigma^\prime)\tilde
p_vdv+\int_s^T\tilde z_vdW^{\mathbb{P}}_v.
\end{equation}
The above BSDE is of the same form as BSDE \eqref{eq:FBSDE:defautable:numerical}, which can also be solved either analytically or numerically.
\begin{remark}
	The decoupled Riccati equations \eqref{FBSDE:defautable:QTSM:Riccati:q}-\eqref{FBSDE:defautable:QTSM:Riccati:boundary} are closely related to the LQ control problem. The existence and uniqueness of solutions to the Riccati equations \eqref{FBSDE:defautable:QTSM:Riccati:q}-\eqref{FBSDE:defautable:QTSM:Riccati:boundary} have been discussed in \citet{Zhou2010}.   We provide a similar proof in the appendix based on the results of \citet{gombani2012arbitrage}.
\end{remark}

\subsubsection{Non-Defaultable case with QTSMs and jumps}
In the framework of QTSMs with jumps, we make the following specifications
\begin{enumerate}[(i)]
	\item $f(s,x)=Ax+B$
	\item $g(s,x)=\Sigma$
	\item $r(x)=x^\prime Qx+R^\prime x+k$
	\item $\lambda(x)=x^\prime L_2x+L_1^\prime x+l$
\end{enumerate}
where $A$ is an $(n\times n)$-matrix of scalars, $B$, $R$ and $L_1$ are $(n\times 1)$-column vectors, $Q$, $\Sigma$ and $L_2$ are $n\times n$ symmetric positive semi-definite matrices, $k$ and $l$ are  scalars. Then the FBSDE \eqref{FBSDE:jumps:forward:a}-\eqref{FBSDE:jumps:backward:a} becomes
\begin{align}
&X_s=X_t+\int_t^s\big(AX_{v-}+B\big)dv+\int_t^s\Sigma dW^{\mathbb{P}}_v+\int_t^s\int_{\mathbf{R}^n}zN^{\mathbb{P}}(dv,dz) \label{FBSDE:jumps:QTSM:forward:a} \\
\notag &Y_s=\int_s^T(X_{v-}^\prime QX_{v-}+R^\prime X_{v-}+k+\frac{1}{2}Z_vZ^\prime_v)dv+\int_s^TZ_vdW^{\mathbb{P}}_v\\
&\quad \quad -\int_s^T\int_{\mathbf{R}^n}(X_{v-}^\prime L_2 X_{v-}+L_1^\prime X_{v-}+L_0)\big(e^{G(v,z)}-1\big)v(dz)dv+\int_s^T\int_{\mathbf{R}^n}G(v,z)N^{\mathbb{P}}(dv,dz). \label{FBSDE:jumps:QTSM:backward:a}
\end{align}
Similar to the result in ATSMs with jumps we obtain the following explicit solution of the FBSDE~(\ref{FBSDE:jumps:QTSM:forward:a})-(\ref{FBSDE:jumps:QTSM:backward:a}).
\begin{theorem}\label{th:JQTSM}
	If the Riccati equation
	\begin{align}
	&\dot{q}_s+q_sA+A^\prime q_s+\frac{1}{2}(q_s^\prime+q_s)\Sigma\Sigma^\prime(q_s^\prime+q_s)
	+\Big[\int_{\mathbf{R}^n}(e^{z^\prime q_sz+u_sz}-1)v(dz)\Big]L_2^\prime-Q=0_{n\times n},  \\
	&\dot{u}_s+u_sA+B^\prime(q_s^\prime+q_s)+u_s\Sigma\Sigma^\prime(q_s^\prime+q_s)
	+\Big[\int_{\mathbf{R}^n}(e^{z^\prime q_sz+u_sz}-1)v(dz)\Big]L_1^\prime-R^\prime=0_{1\times n},\\
	&q_T=0,\quad u_T=0
	\end{align}
	admits  unique bounded solutions $q(\cdot)$, $u(\cdot)$ over the interval $[0,T]$, then the FBSDE
	\eqref{FBSDE:jumps:QTSM:forward:a}-\eqref{FBSDE:jumps:QTSM:backward:a} admits a unique solution
	and $(Y,Z,G)$ has explicit expression in terms of $X$ as follows 
	\begin{align*}
	&Y_s= -(X_s^\prime q_sX_s+u_sX_s+p_s), \\
	&Z_s=\Big(X_{t-}^\prime(q_s+q_s^\prime)+u_s\Big)\Sigma, \quad \text{and} \\
	&G(s,z)=z^\prime q_sz+u_sz,
	\end{align*}
	where $p_s$ is given by
	\begin{equation*}
	p_s=-\int_s^T\left(k-L_0\Big[\int_{\mathbf{R}^n}(e^{z^\prime
		q_vz+u_vz}-1)v(dz)\Big]-u_vB-\frac{1}{2}\text{tr}\big((q_v+q_v^\prime)\Sigma\Sigma^\prime\big)-\frac{1}{2}u_v\Sigma\Sigma^\prime
	u_v^\prime\right)dv.
	\end{equation*}
\end{theorem}
We omit the proof of Theorem~\ref{th:JQTSM} as it is similar to the proof of Theorem~\ref{theorem:BSDE:jump:solution}.

We now consider the EMT problems for futures and forward prices.

\subsection{Futures and forward prices}\label{futures and forward}
Suppose the factor process $X_s$ given by equation (\ref{eq:SDE:forward:a}) drives not only the short rate but also a risky asset price.  We assume that the risky asset price is a function of factors, $S_s=S(s,X_s)$, for some function $S(\cdot,\cdot): [0,\infty)\times \mathbf{R}^n\to (0,\infty)$. For instance,  $S(\cdot,\cdot)$ can be specified by
$$S(s,x)=e^{A_s^\prime x+h_s},$$
which we refer to as an affine price model (APM), or
$$S(s,x)=e^{x^\prime B_s x+A_s^\prime x+h_s}$$
which we refer to as a quadratic price model (QPM), where $B_s: [0,T]\to \mathbf{R}^{n\times n}, A_s:[0,T]\to \mathbf{R}^n, h_s: [0,T]\to \mathbf{R}$. 

We next consider futures and forward contract on the risky asset $S$ and associate the futures prices and forward prices with EMT problems.

\subsection{Futures prices}
The futures price of the risky asset $S$ is given by
\begin{equation}\label{eq:future:price:a}
G(t,T)=\mathbb{E}_{\mathbb{P}}[S(T,X_T)|\mathcal{F}_t], 
\end{equation}
at time $t$ for maturity $T$, %
and let
\begin{equation}\label{problem:EMT:future:a}
\left\{
\begin{aligned}
dX_s&=f(s,X_s)ds+g(s,X_s)dW_s^{\mathbb{P}}\\
V^G_{t,T}&=\inf_{\mathbb{Q}^G\in \mathcal{P}_t(1)}\big\{\mathbb{E}_{\mathbb{Q}^G}[-\ln S(T,X_T)|\mathcal{F}_t]+H_{t,T}(\mathbb{Q}^G|\mathbb{P})\big\}.
\end{aligned}
\right.
\end{equation}

By Proposition \ref{proposition:duality:relation:a} the solution of the EMT Problem~(\ref{problem:EMT:future:a}) is given by the optimal measure $\mathbb{Q}^{G^\star}$,  that is determined by
\begin{align}\label{eq:EMT:future:optimal:measure:a}
\left.\frac{d \mathbb{Q}^{G^\star}}{d \mathbb{P}}\right|_{\mathcal{F}_T}&=\frac{S(T,X_T)}{\mathbb{E}_{\mathbb{P}}[S(T,X_T)|\mathcal{F}_t]},\\
\label{eq:EMT:future:valuefunction:a}
V_{t,T}^G &=-\ln \{\mathbb{E}_{\mathbb{P}}[S(T,X_T)|\mathcal{F}_t]\}. 
\end{align}
Equation (\ref{eq:EMT:future:valuefunction:a}) connects the EMT Problem~(\ref{problem:EMT:future:a}) with the futures price as
$$
V^G_{t,T}=-\ln G(t,T).
$$
This relationship allows us to give the following financial interpretation.  

\begin{remark}
	Suppose a financial agent holds a long position in a futures contract on the risky asset $S$ at time t with a futures price $c$. If the risky asset has a price $S(T,X_T)$ at the expiration time $T$, by marking to market through the time period $[t,T]$, the logarithmic return on the investment is
	$$\gamma=\ln \frac{S(T,X_T)}{c}.$$
	Then the EMT Problem~(\ref{problem:EMT:future:a})  is equivalent to 
	\begin{equation}\label{eq:EMT:future:interpretation:a}
	\ln\frac{G(t,T)}{c}=\sup_{\mathbb{Q}\in \mathcal{P}_t(1)}\big\{\mathbb{E}_{\mathbb{Q}^G}[\gamma|\mathcal{F}_t]-H_{t,T}(\mathbb{Q}^G|\mathbb{P})\big\}.
	\end{equation}
	The right-hand side  of equation (\ref{eq:EMT:future:interpretation:a}) maximizes the logarithmic return $\gamma$ under  $\mathbb{Q}^{G^{*}}$ with an entropy penalty term for removing the market risk the futures contract caused by the volatility risk of underlying risky asset.
\end{remark}
By Theorem~\ref{thrm_main}, the measure $\mathbb{P}^{\star}$ is characterized by the following decoupled FBSDE
\begin{align}
X_s&=X_t+\int_t^s f(v,X_v)dv+\int_t^s g(v,X_v)dW_v^\mathbb{P}, \label{FBSDE:future:forward:a}\\
Y_s&=-\ln [S(T,X_T)]-\int_s^T\frac{1}{2}Z_vZ_v^\prime dv+\int_t^TZ_vdW_v^\mathbb{P}.\label{FBSDE:future:backward:a}\\
	 Y_s &\triangleq \ln \left\{\mathbb{E}_{\mathbb{P}}[e^{-\int_{t}^{T}r(X_u)du}\cdot C_T |\mathcal{F}_s]\right\}.\nonumber
\end{align}
If the above FBSDE admits a solution triple $(X,Y,Z)$, then the value function and the measure $\mathbb{P}^{\star}$ is defined by
\begin{align*}
V^G_{t,T}&=Y_t,\\
\left.\frac{d \mathbb{Q}^{G^\star}}{d \mathbb{P}}\right|_{\mathcal{F}_T}&=e^{-\int_t^T\frac{1}{2}Z_vZ_v^\prime dv+\int_t^T Z_vdW^\mathbb{P}_v}.
\end{align*}
\citet{hyndman2009forward} and \citet{Zhou2010} studied the 
the FBSDE (\ref{FBSDE:future:forward:a})-(\ref{FBSDE:future:backward:a}) in the framework of ATSMs and QTSMs, respectively, and gave explicit solutions.

We next consider a forward contract on the risky asset. 
\subsection{Forward prices}
The forward price of the risky asset $S$ is given by
\begin{equation}\label{eq:forward:price}
F(t,T)=\frac{\mathbb{E}_{\mathbb{P}}[e^{-\int_t^Tr(X_v)dv}S(T,X_T)|\mathcal{F}_t]}{P(t,T)}, 
\end{equation}
at time $t$ for maturity $T$. To ensure that the forward price is not simply equal to the futures price we assume that the interest rate process is stochastic and the factors influencing the interest rate are not independent of the factors influencing the underlying asset price.  Further, to preclude the case where the numerator of equation~(\ref{eq:forward:price}) reduces to the underlying asset price at time $t$ we suppose that the asset pays a stochastic dividend or convenience yield.  

Similar to the derivation of  the EMT Problem in Section~\ref{s_EMT_if} we let $\varphi=(\ln S(T,X_T)-\int_{t}^{T}r_vdv)$ and associate the forward price with the following EMT problem
\begin{equation}\label{problem:EMT:forward}
\left\{
\begin{aligned}
dX_s&=f(s,X_s)ds+g(s,X_s)dW_s^{\mathbb{P}}\\
V^F_{t,T}&=\inf_{\mathbb{Q}^F\in \mathcal{P}_t(1)}\big\{\mathbb{E}_{\mathbb{Q}^F}\big[-\ln S(T,X_T)+\int_t^Tr_vdv|\mathcal{F}_t\big]+H_{t,T}(\mathbb{Q}^F|\mathbb{P})\big\}.
\end{aligned}
\right.
\end{equation}
By Proposition \ref{proposition:duality:relation:a} the solution to the EMT Problem~(\ref{problem:EMT:forward}) is given by the optimal measure $\mathbb{Q}^{F^\star}$, that is determined by
\begin{equation}
\left.\frac{d \mathbb{Q}^{F^\star}}{d \mathbb{P}}\right|_{\mathcal{F}_T}=\frac{S(T,X_T)e^{-\int_t^Tr_vdv}}{\mathbb{E}_{\mathbb{P}}[S(T,X_T)e^{-\int_t^Tr_vdv}|\mathcal{F}_t]},
\end{equation}
and the optimal value function  given by
\begin{equation}\label{eq:EMT:forward:valuefuncation}
V_{t,T}^F =-\ln\big(\mathbb{E}_{\mathbb{P}}[e^{-\int_t^Tr(X_v)dv}S(T,X_T)|\mathcal{F}_t]\big).
\end{equation}
Equation (\ref{eq:EMT:forward:valuefuncation}) connects the EMT Problem~(\ref{problem:EMT:forward}) with the forward price as
$$V^F_{t,T}=-\ln\big(F(t,T) P(t,T)\big).$$
Therefore, we have the following financial interpretation of the EMT Problem~(\ref{problem:EMT:forward}).

\begin{remark}
	Suppose a financial agent enters into a forward agreement to receive the asset at time $T$ but pays $c$ at time $t$. At the expiration time $T$ the agent receives $S(T,X_T)$ and the  logarithmic return on the investment over the time period $[t,T]$ is
	$$\gamma=\ln \frac{S(T,X_T)}{c}.$$
	The excess return over the risk-free rate, $\tilde{\gamma}$, is given by
	$$\tilde\gamma=\gamma-\int_{t}^{T}r(X_v)dv.$$
	Then the EMT Problem~(\ref{problem:EMT:forward})  is equivalent to 
	\begin{equation}\label{eq:EMT:forward:interpretation}
	\ln\frac{F(t,T)P(t,T)}{c}=\sup_{\mathbb{Q}^F\in
		\mathcal{P}_t(1)}\big\{\mathbb{E}_{\mathbb{Q}^F}[\tilde\gamma|\mathcal{F}_t]-H_{t,T}(\mathbb{Q}^F|\mathbb{P})\big\}. 
	\end{equation}
	Similar to the financial interpretation of the EMT Problem for the bond and futures contract the right hand side of equation (\ref{eq:EMT:forward:interpretation}) maximizes the excess  return $\tilde \gamma$ under  $\mathbb{Q}^{F^{*}}$ with an entropy penalty term for removing the market risk of the value of the forward commitment due to the volatility risk of the factor process that determines both the interest rate and underlying asset volatilities.
\end{remark}
Using Theorem~\ref{thrm_main}, we characterize the EMT Problem \ref{problem:EMT:future:a}  by the FBSDE
\begin{align}
X_s&=X_t+\int_t^s f(v,X_v)dv+\int_t^s g(v,X_v)dW_v^\mathbb{P} \label{FBSDE:forward:forward}\\
Y_s&=-\ln [S(T,X_T)]+\int_s^T[r(X_v)-\frac{1}{2}Z_vZ_v^\prime] dv+\int_s^TZ_vdW_v^\mathbb{P} \label{FBSDE:forward:backward}
\end{align}
If the above FBSDE admits a solution triple $(X,Y,Z)$, then the value function and the optimal measure to the
EMT may be expressed as
\begin{align*}
V^F_{t,T}&=Y_t,\\
\left.\frac{d \mathbb{Q}^{F^\star}}{d \mathbb{P}}\right|_{\mathcal{F}_T}&=e^{-\int_t^T\frac{1}{2}Z_vZ_v^\prime dv+\int_t^T Z_vdW^\mathbb{P}_v}.
\end{align*}
\citet{hyndman2009forward} and  \citet{Zhou2010} also studied the 
the FBSDE (\ref{FBSDE:forward:forward})-(\ref{FBSDE:forward:backward}) in the framework of ATSMs and QTSMs, respectively, and gave explicit solutions.

The EMT approach seems to be more flexible with respect to the dynamics of the factors process than the OSC approach. In next section we extend the EMT approach to include jumps in the factors which would be difficult to incorporate using the OSC approach.

The next section compares the entropic measure transform problem with the optimal stochastic control 
problem proposed by \citet{gombani2012arbitrage}. It is found that there exists an equivalence between these two approaches, within the scope of the term-structure of interest.  

\subsection{Equivalence between the EMT problem and the OSC problem in Bond Pricing}
Following \citet{gombani2012arbitrage}, considered the bond pricing problem under the same general framework as we set up in previous section. To avoid confusion, we denote the factor process by $\tilde X_s$ in the context of  \citet{gombani2012arbitrage}. Additionally, $\tilde X_s$ is assumed to solve the SDE
\begin{equation}\label{eq:OSC:SDE}
d\tilde X_s=\big[f(s,\tilde X_s)+g(s,\tilde X_s)u_s^\prime\big]ds+ g(s,\tilde X_s)dW^{\mathbb{P}}_s; X_t=x.
\end{equation}
so that the price of default-free bond, denoted by $P(t,T,x)$, at time $t$ is given by
$$P(t,T,x)=\mathbb{E}_{\mathbb{P}}[e^{-\int_{t}^{T}r_vdv} |\mathcal{F}_t]=\mathbb{E}_{\mathbb{P}}[e^{-\int_{t}^{T}r_vdv} |X_t=x].$$
Assuming $P(t,T,x)\in \mathcal{C}^{1,2}$, a sufficient condition for the term-structure induced by $P(t,T,x)$ to be arbitrage-free is that $P(t,T,x)$ satisfies the following partial differential equation (see \citet[Proposition 21.2]{bjork2004arbitrage})
\begin{equation}
\left\{
\begin{aligned}
& \frac{\partial}{\partial t}P(t,T,x)+f^\prime(t,x)\nabla_xP(t,T,x)+\frac{1}{2}\text{tr}\big(g^\prime(t,x)\nabla_{xx}P(t,T,x)
g(t,x)\big)-P(t,T,x)r(t,x)=0\\
& P(T,T,x)=1.
\end{aligned}
\right. \label{eq:term:structure:PDE:a}
\end{equation}

\citet{gombani2012arbitrage} transform equation~(\ref{eq:term:structure:PDE:a}) to an equivalent Hamilton-Jacobi-Bellman equation which corresponds to the following optimal stochastic control (OSC) problem.
\begin{problem} \label{problem:OSC:a}
	On a filtered probability space $(\Omega,\mathcal{F},\{\mathcal{F}_s,0\leq s\leq T\},\mathbb{P})$, with a Markovian process $\tilde X_s$ given by equation~\eqref{eq:EMT:SDE}.  
	
	Let $\mathcal{U}$ be the admissible control set, then for any control $u\in\mathcal{U}$ and
	$t\in[0,T]$, consider a performance criterion 
	$\tilde{J}_{t,T}(u)$ of the form
	\begin{equation}
	\tilde{J}_{t,T}(u)=\mathbb{E}_{\mathbb{P}}^{t,x}\left[\int_t^T\big(\frac{1}{2}u_vu_v^\prime+r(\tilde X_v)\big)dv\right],
	\end{equation}
	where $\mathbb{E}_{\mathbb{P}}^{t,x}$ denotes the conditional expectation given $\tilde X_t=x$.  
	The optimal control problem is  
	$$W_{t,T}=\inf_{u\in\mathcal{U}}\tilde{J}_{t,T}(u).$$
\end{problem}
\citet{gombani2012arbitrage}  established a connection between the price of default-free bonds and
the OSC Problem~\ref{problem:OSC:a} by showing that
\begin{equation*}
P(t,T,x)=e^{-W_{t,T}(x)}.
\end{equation*}

We next explore an equivalence relationship between the EMT problem and the OSC problem.  For any
$\mathbb{Q}\in \mathcal{P}_t(1)$,
the Radon-Nikodym derivative process is of the following form 
\begin{equation*}
\left.\frac{d \mathbb{P}^{\star}}{d \mathbb{P}}\right|_{\mathcal{F}_s}=
\begin{cases}
1, & \quad 0\leq s\leq t.\\
\Lambda_s, & \quad t<s\leq T.
\end{cases}
\end{equation*}
where $\Lambda_s$ is an $(\mathcal{F},\mathbb{P})$-martingale from $t$ to $T$. Since $\Lambda_s$ is positive almost surely,   by the martingale representation theorem,
there exists an $\mathcal{F}$-predictable $(1\times n)$-vector process  $u$ such that 
\begin{equation}
\left.\frac{d \mathbb{P}^{\star}}{d \mathbb{P}}\right|_{\mathcal{F}_s}=e^{-\int_t^s\frac{1}{2}u_vu_v^\prime
	dv+\int_t^s u_vdW^\mathbb{P}_v},  \quad t< s\leq T
\label{eq:denstiy:process:a}
\end{equation}
where $u$ is an $\mathcal{F}$-predictable $(1\times n)$-vector process. In the remaining part of this
section we denote by $\mathbb{Q}^u$ the probability measure associated with the density process in equation
(\ref{eq:denstiy:process:a}).
Then, by Girsanov's theorem, the process $W^{\mathbb{Q}^u}$ defined as
$$W^{\mathbb{Q}^u}_s=W^\mathbb{P}_s-\int_t^su^\prime_vdv,\quad t< s\leq T$$
is a Brownian motion under $\mathbb{Q}^u$.
Then we calculate the relative entropy of $\mathbb{Q}^u$
with respect to $\mathbb{P}$ explicitly in terms of $u$ as follows
\begin{align}
\notag
H_{t,T}(\mathbb{Q}^u|\mathbb{P})&=\mathbb{E}_{\mathbb{Q}^u}[\ln(\frac{d\mathbb{Q}^u}{d\mathbb{P}})|\mathcal{F}_t]\\
\notag &=\mathbb{E}_{\mathbb{Q}^u}[\Big(-\int_t^T\frac{1}{2}u_vu_v^\prime dv+\int_t^Tu_vdW_v^\mathbb{P}\Big)|\mathcal{F}_t]\\
\notag &=\mathbb{E}_{\mathbb{Q}^u}[\Big(\int_t^T\frac{1}{2}u_vu_v^\prime
dv+\int_t^TZ_vdW^{\mathbb{Q}^u}_v\Big)|\mathcal{F}_t]\\
&=\mathbb{E}_{\mathbb{Q}^u}[\int_t^T\frac{1}{2}u_vu_v^\prime dv|\mathcal{F}_t]. \label{eq:entropy:explicit:b}
\end{align}

Substituting the explicit expression of the relative entropy in equation (\ref{eq:entropy:explicit:b}) into 
\begin{equation}\label{equation:EMT:performance:a}
J_{t,T}(\mathbb{Q})=\mathbb{E}_{\mathbb{Q}}\Big[\int_{t}^{T}r(X_v)dv
\Big|\mathcal{F}_t\Big]+H_{t,T}(\mathbb{P}^{\star}|\mathbb{P}),
\end{equation} 
we restate the EMT Problem as follows
\begin{problem}\label{problem:EMT:b}
	On a filtered probability space $(\Omega,\mathcal{F},\{\mathcal{F}_s,0\leq s\leq T\}, \mathbb{P})$ suppose that  the factor process $(X_s, 0\leq s\leq T)$ is given by
	\begin{equation}
	dX_s=f(s,X_s)ds+g(s,X_s)dW^\mathbb{P}_s.
	\label{eq:EMT:SDE}
	\end{equation}
	Find the optimal measure $\mathbb{Q}^\star\in \mathcal{P}_t(1)$ such that
	\begin{equation}
	V_{t,T}=J_{t,T}(\mathbb{Q}^\star)=\inf_{\mathbb{Q}^u\in
		\mathcal{P}_t(1)}\mathbb{E}_{\mathbb{Q}^u}\Big[\int_{t}^{T}\big(r(X_v)+\frac{1}{2}u_vu_v^\prime\big) dv\Big].
	\end{equation} 
\end{problem}

In the OSC Problem \ref{problem:OSC:a}, the distribution of $\tilde X_s$ is changed by the control process $u$. In
the EMT Problem \ref{problem:EMT:b}, the distribution of $X_s$ is subject to the measure transformation from
$\mathbb{P}$ to $\mathbb{Q}^u$. Note that $\tilde X_s$ in equation (\ref{eq:OSC:SDE}) and $X_s$ in equation
(\ref{eq:EMT:SDE}) follow SDEs of the same form under different measures, in other words, the $u$ controlled
process $\tilde X_s$ has the same distribution under $\mathbb{P}$ as the process $X_s$ does under
$\mathbb{Q}^u$. Hence for each admissible control $u$ in the OSC problem with performance functional
$\tilde{J}_{t,T}(u)$, there exists a corresponding measure $\mathbb{Q}^u$ in the EMT problem with performance
functional $J_{t,T}(\mathbb{Q}^u)$, and $\tilde{J}_{t,T}(u)=J_{t,T}(\mathbb{Q}^u)$. So the optimal control $u^\star$
also corresponds to the entropic measure transform $\mathbb{Q}^\star=\mathbb{Q}^{u^\star}$. In that sense, the OSC problem is equivalent to the EMT problem. 

\begin{example}
	Now we compare the EMT problem and the OSC problem  under the framework of QTSMs with specifications
	\begin{enumerate}[(i)]
		\item $f(s,x)=Ax+B$
		\item $g(s,x)=\Sigma$
		\item $r(x)=x^\prime Qx+R^\prime x+k$
	\end{enumerate}
	where $A$ is an $(n\times n)$-matrix of scalars, $B$ and $R$ are $(n\times 1)$-column vectors, $Q$ and
	$\Sigma$ are $n\times n$ symmetric positive semi-definite matrices, $k$ is a scalar. The OSC
	Problem \ref{problem:OSC:a}  becomes
	\begin{equation}	\label{problem:OSC:QTSM:b}	
	\left\{\begin{aligned}
	d\tilde X_s&=\big(A\tilde X_s+B+\Sigma u^\prime_s\big)ds+ \Sigma dW^{\mathbb{P}}_s,\\
	V_{t,T}&=\inf_{u\in\mathcal{U}}\tilde{J}_{t,T}(u)=\inf_{u\in\mathcal{U}}\mathbb{E}_{t,x}[\int_t^T\big(\tilde X_v^\prime Q\tilde X_v+R^\prime \tilde X_v+k+\frac{1}{2} u_v u_v^\prime\big)dv].
	\end{aligned}
	\right.
	\end{equation}

	The OSC Problem	\ref{problem:OSC:QTSM:b} is actually a linear-quadratic-Gaussian (LQG) control problem, whose optimal control $u^\star_s$ is of feedback form  (see \citet[Proposition 3.4]{gombani2012arbitrage}) 
	\begin{equation}\label{eq:OSC:QTSM:opt:control}
	u^\star_s=u^\star(s,\tilde X_s)=\Big(X_s^\prime(q_s+q_s^\prime)+v_s\Big)\Sigma, \quad t\leq s\leq T 
	\end{equation}
	with the value function $W_{t,T}(x)$ given by
	\begin{equation}\label{eq:OSC:QTSM:valuefunction}
	W_{t,T}(x)=x^\prime q_tx+ v_tx+p_t, 
	\end{equation}
	where $q_s, v_s, p_s$ satisfy the following ODE system
	\begin{equation}\label{ODEs:OSC}
	\begin{cases}
	\dot{q}_s+A^\prime q_s+q_sA-2q_s\Sigma \Sigma^\prime q_s+Q=0\\
	\dot{v}_s+v_sA+2B^\prime q_s^\prime-2v_s\Sigma \Sigma^\prime q^\prime_s+R=0\\
	\dot{p}_s+v_sB+\text{tr}(\Sigma^\prime q_s\Sigma)-\frac{1}{2}v\Sigma\Sigma^\prime v^\prime_s+k=0\\
	q_T=0,\quad v_T=0,\quad p_T=0.
	\end{cases} 
	\end{equation}
	
	Under the framework of QTSMs the EMT Problem is specified as
	\begin{equation}\label{problem:EMT:c}
	\left\{\begin{aligned}
	dX_s&=\big(AX_s+B\big)ds+\Sigma dW_s^\mathbb{P},\\
	V_{t,T}&=\inf_{\mathbb{Q}\in\mathcal{P}_t(1)}\mathbb{E}_{\mathbb{Q}}[\int_t^T\big(X_v^\prime QX_v+R^\prime X_v+k\big)dv|\mathcal{F}_t]+H_{t,T}(\mathbb{P}^{\star}|\mathbb{P}).
	\end{aligned}
	\right.  
	\end{equation}
	From Corollary~\ref{cor_bond_pricing}, we know  the EMT Problem~(\ref{problem:EMT:c}) is completely characterized via the related FBSDE
	\begin{align}
	X_s&=X_t+\int_t^s \left(AX_v+B+\Sigma Z_v^\prime\right)dv+\int_t^s\Sigma dW_v^\mathbb{P} 	\label{FBSDE:QTSM:forward:a}\\
	Y_s&=\int_s^T(X_v^\prime QX_v+R^\prime X_v+k-\frac{1}{2}Z^\prime_vZ_v)dv+\int_s^TZ_vdW_v^\mathbb{P}.
	\label{FBSDE:QTSM:backward:a}  
	\end{align}
	The value function is given by
	\begin{equation}
	V_{t,T}=Y_t,
	\end{equation}
	and the entropic measure transform is determined by
	\begin{equation}
	\left.\frac{d \mathbb{Q}^\star}{d \mathbb{P}}\right|_{\mathcal{F}_T}=e^{-\int_t^T\frac{1}{2}Z_vZ_v^\prime dv+\int_t^T Z_vdW^\mathbb{P}_v}.
	\end{equation}

	\citet{Zhou2010} proved that the FBSDE (\ref{FBSDE:QTSM:forward:a})-(\ref{FBSDE:QTSM:backward:a}) admits a unique solution $(X,Y,Z)$, and $(Y, Z)$ has  explicit expressions in terms of $X$ 
	\begin{align*}
	Y_s&= X_s^\prime q_sX_s+v_sX_s+p_s,\\
	Z_s&=\Big(X_s^\prime(q_s+q_s^\prime)+v_s\Big)\Sigma,
	\end{align*}
	where $q_s, v_s, p_s$ satisfy the same ODE system (\ref{ODEs:OSC}). Not surprisingly, the Girsanov kernel $Z_s$ for the transition from  
	$\mathbb{P}$ to $\mathbb{Q}^\star$ is the same as  the optimal control $u^\star$, i.e. $Z_s=u^\star_s$,
	and they give the same value function $V_{t,T}=W_{t,T}$.
\end{example}

An example of the numerical implementation of the EMT method is considered in the setting of defaultable bonds.  

\section{Numerical Illustration}\label{s_Numerical}
We consider a one dimensional factor process $X$ satisfying 
\begin{equation*}
d X_t = (aX_t +b) dt + \sigma \sqrt{\alpha +\beta X_v}dW^{\mathbb{P}}_t.
\end{equation*}
The interest rate is given by
\begin{equation*}
r(X_t) = RX_t+k.
\end{equation*}
We suppose the underlying company value $V$ satisfies 
\begin{equation*}
V_t = V_0 \exp\{ \int_0^t (r(X_v)  - \frac{1}{2}\sigma_V^2)dv + \sigma_V W^{\mathbb{P}}_t \}.
\end{equation*}
Default is triggered if the value process $V$ crosses below a certain level $\kappa V_0$, i.e.
\begin{equation}
\tau:=\inf\{t\geq 0, V_{t}\leq \kappa V_{0} \}.
\end{equation} 
Then the random payoff $C_T$ is given by 
$$C_T= \xi \cdot 1_{\tau \leq T}+ 1_{\tau > T}$$
where  $\xi$ is the recovery rate in case of default.

The price of the defaultable bond is given by 
\begin{equation*}
D(t,T)=\mathbb{E}_{\mathbb{P}}[e^{-\int_{t}^{T}(RX_v+k)dv}\cdot C_T |\mathcal{F}_t].
\end{equation*}
The solution to the associated EMT problem is characterized by the FBSDE
\begin{align}
&X_t=X_0+\int_0^t \left(aX_v+b\right)dv+\int_0^t\sigma\sqrt{\alpha+\beta X_v}dW^{\mathbb{P}}_v \label{FBSDE:defautable:ATSM:forward:aa}\\
&Y_t=-\ln C_T+\int_t^T(R X_v+k-\frac{1}{2}Z_v^2)dv+\int_t^TZ_vdW^{\mathbb{P}}_v. \label{FBSDE:defautable:ATSM:backward:aa}
\end{align}
We have explicit expression for the solution to FBSDE \eqref{FBSDE:defautable:ATSM:forward:aa}-\eqref{FBSDE:defautable:ATSM:backward:aa}
\begin{align}
Y_t&= -(U_tX_t+p_t),\quad \label{eq:FBSDE:defautable:rep:Yy}	\\
Z_t&=\sigma U_t(\sqrt{\alpha+\beta X_t})+q_{t}, \label{eq:FBSDE:defautable:rep:Zz}	
\end{align}
where $U_s$ satisfies the Riccati equation
\begin{align}
&\dot{U}_t+aU_t+\frac{\beta}{2}\sigma^2 U_t^2- R=0, \quad t\in[0,T] \label{eq:Riccati:defautablee}\\
&U_T=0, \label{eq:Riccati:defaultable:boundaryy}
\end{align}
and $(p, q)$ solves the BSDE
\begin{equation}\label{eq:BSDE_numerical_a}
p_t= -\ln C_T-\int_t^T\left(k-\frac{\alpha}{2}\sigma^2 U_v^2-bU_v-\frac{1}{2}q_v^2\right) dv-\int_t^Tq_vdW^{\mathbb{P}}_v.  
\end{equation}
The defaultable bond price can be expressed as
\begin{equation}
D(t,T) = \exp\{-Y_t\}.
\end{equation}
The aggregate relative entropy of the optimal measure $\mathbb{Q}^\star$ with respect to $\mathbb{P}$ is given by 
$$H_{t,T}(\mathbb{Q}^\star|\mathbb{P})=\mathbb{E}^{\mathbb{Q}^\star}[\int_t^T\frac{1}{2}Z_v^2 dv|\mathcal{F}_t].$$

We introduce the following proposition which gives explicit solution to a special type of quadratic BSDEs.
\begin{proposition}\label{prop:quadratic_BSDE_explicit_solution}
	On a probability space $(\Omega, \mathcal{F},\{\mathcal{F}_t,t\geq0\},\mathbb{P}$, consider the following BSDE
	$$y_t=\xi-\int_t^T(\frac{1}{2}z_s^\prime z_s+g_s)ds-\int_t^Tz_sdW^\mathbb{P}_s,$$
	where $(y_t,z_t)\in \mathbf{R}\times \mathbf{R}^n$, $\xi$ is real-valued $\mathcal{F}_T$-measurable random variable, $g_t$ is real-valued $\mathcal{F}_t$-adapted process satisfying
	$\mathbb{E}_{\mathbb{P}}[\sup_{0\leq t\leq T}|g_t|^2]<\infty$.
	Then $y_t$ can be expressed explicitly as
	$$y_t=-\ln\{\mathbb{E}_{\mathbb{P}}[e^{-\xi}|\mathcal{F}_t]\}-\int_0^tg_sds.$$
\end{proposition}
\begin{proof}
	Make the exponential transformation $\tilde y_t=e^{-y_t}$, by It\^{o}'s formula  $\tilde {y}_t$ satisfies
	
	$$\tilde {y}_t=e^{-\xi}+\int_t^Tg_s\tilde {y}_sds+\int_t^T\tilde {y}_sz_sdW^\mathbb{P}_s.$$
	Define the adjoint process $$x_s=e^{\int_t^sg_udu}, \quad s\geq t.$$
	Notice that $x_t=1$, and apply  It\^{o} formula to $x_s\cdot \tilde {y}_s$ from $t$ to $T$, to find
	\begin{align}
	\notag \tilde {y}_t & =x_Te^{-\xi}+\int_t^T\tilde {y}_sx_sz_sdW^\mathbb{P}_s \\
	& =e^{\int_t^Tg_sds-\xi}+\int_t^T\tilde {y}_se^{\int_t^sg_udu}z_sdW^\mathbb{P}_s.\label{2.6}
	\end{align}
	Take conditional expectation on $\mathcal{F}_t$ of both sides of \eqref{2.6}, we obtain
	\begin{align*}
	\tilde y_t & =\mathbb{E}_{\mathbb{P}^T}[e^{\int_t^Tg_sds-\xi}|\mathcal{F}_t] \\
	& =e^{\int_t^T g_sds}\mathbb{E}_{\mathbb{P}}[e^{-\xi}|\mathcal{F}_t].
	\end{align*}
	Finally we have
	\begin{align*}
	y_t & =-\ln{\tilde y_t}\\
	& =-\ln\{\mathbb{E}_{\mathbb{P}}[e^{-\xi}|\mathcal{F}_t]\}-\int_0^tg_sds.
	\end{align*}
	
\end{proof}

\begin{remark}
	The existence and uniqueness of the solution to general quadratic BSDEs was proven by  \citet{kobylanski2000backward}.
	Proposition \eqref{prop:quadratic_BSDE_explicit_solution} is only a special case in which we can give the explicit solution.
\end{remark}

Applying Proposition \ref{prop:quadratic_BSDE_explicit_solution} to the BSDE \eqref{eq:BSDE_numerical_a},  we may express $p_t$ explicitly as
\begin{equation}
p_t=-\ln\{\mathbb{E}_{\mathbb{P}}[\frac{1}{C_T}|\mathcal{F}_t]\}-\int_0^t\left(k-\frac{1}{2}U_vk_0U_v^\prime-U_vB\right)dv.
\end{equation}
However, we do not have an explicit expression for the process $q_t$.  Alternatively we can solve BSDE
\eqref{eq:BSDE_numerical_a}  numerically.  We can transform the quadratic BSDE \eqref{eq:BSDE_numerical_a}  into an
equivalent linear BSDE by defining
$$	\tilde p_t=e^{-p_t}, \quad	\tilde q_t=\tilde p_t\cdot q_t,$$
so that the BSDE \eqref{eq:BSDE_numerical_a} is equivalent to
\begin{equation}\label{eq_BSDE_linear}
\tilde p_t=\frac{1}{C_T}+\int_t^T(k-\frac{\alpha}{2}\sigma^2 U_s^2-bU_s)\tilde p_sds+\int_t^T\tilde
q_sdW_s.
\end{equation}

We approximate the solution to the BSDE (\ref{eq_BSDE_linear}) by considering the following discretized BSDE
\begin{align*}
\tilde p_{t_{m+1}} &= \tilde p_{t_{m}} -(k-\frac{\alpha}{2}\sigma^2 U_{t_m}^2-bU_{t_m})\tilde p_{t_m}\Delta t - \tilde q_{t_{m}}\Delta W_{t_m}^{\mathbb{P}}, \quad t_0 \leq t_m \leq t_M, \\
\tilde p_{t_M} &= \frac{1}{C_T}.
\end{align*}
The discretized BSDE can be solved using the following recursive scheme (see \cite{gobet2005regression})
\begin{align*}
\tilde q_{t_m} & = \frac{1}{\Delta t} \mathbb{E}[p_{t_{m+1}}\Delta W^{\mathbb{P}}_{t_m}| \mathcal{F}_{t_m}],\\
\tilde{p}_{t_m} &= \frac{\mathbb{E}[\tilde{p}_{t_{m+1}}|\mathcal{F}_{t_m}]}{1-(k-\frac{\alpha}{2}\sigma^2 U_{t_m}^2-bU_{t_m})\Delta t}.
\end{align*}
We estimate the conditional expectation by the Monte-Carlo regression approach proposed by \cite{gobet2005regression}. 
With a time discretization over $[0,T]$ we use the Euler scheme  to generate the paths of the forward process $X_t$ in \eqref{FBSDE:defautable:ATSM:forward:aa}, approximated by $X_{t_m}$. We denote by $U_{t_m}$ the numerical solution to the Riccati equation \eqref{eq:Riccati:defautablee}.  Then the defaultable bond price is estimated as $$D(t_{t_m},T) \approx \exp(U_{t_m}X_{t_m}+p_{t_m}).$$
The aggregate relative entropy of the optimal measure $\mathbb{Q}^\star$ with respect to $\mathbb{P}$ is estimated as
$$H_{t_{m},T}(\mathbb{Q}^\star|\mathbb{P})=\mathbb{E}^{\mathbb{Q}^\star}[\sum_{t\leq t_m \leq T}\frac{1}{2}\left(\sigma U_{t_m}\left(\sqrt{\alpha+\beta X_{t_m}}+q_{t_m}\right)\right)^2 \Delta t |\mathcal{F}_t].$$

Consider the parameters $a = -1\times 10^{-2}, b = 1 \times 10^{-5}, \sigma = 7.4\times 10^{-3}, R = 1, k =0, T = 1, V_0 = 20, \sigma_V = 0.2, \kappa = 0.8$ and $\xi$ (recovery rate) is a uniform random variable on $[0.4, 0.6]$.   Figure \ref{fig:IRProcess} shows one sample path of the realized interest rate process. Figure \ref{fig:DefautableBondPriceWithDefault} presents the case where default occurs before the maturity $T$ as the value process crosses the default barrier. Figure \ref{fig:DefautableBondPriceWithDefault} also shows the evolution of the defaultable bond price. The defaultable bond price fluctuates more before the default time, which is affected not only by the distance between the value process and the default barrier but also the time to maturity.  The defaultable bond price after default time is almost constant which is determined by the recovery rate.  Lastly, Figure \ref{fig:DefautableBondPriceWithDefault} illustrates the 
the aggregate relative entropy process $H(t,T)$.  Similar to the price process, the aggregate relative entropy process fluctuates more before default due to uncertainty of default timing. After default, the aggregate relative entropy decreases to zero almost linearly since the major uncertainty after default comes from the interest rate process which is negligible compared with default risk. Figure \ref{fig:DefautableBondPriceWithoutDefault} illustrates the case where default does not occur before maturity. The default bond price fluctuates strongly in the early period of horizon $[0,T]$ and then converges to 1 as time approaches maturity without occurrence of default.

\begin{figure}[!]
	\centering
	\includegraphics[width=0.7\linewidth]{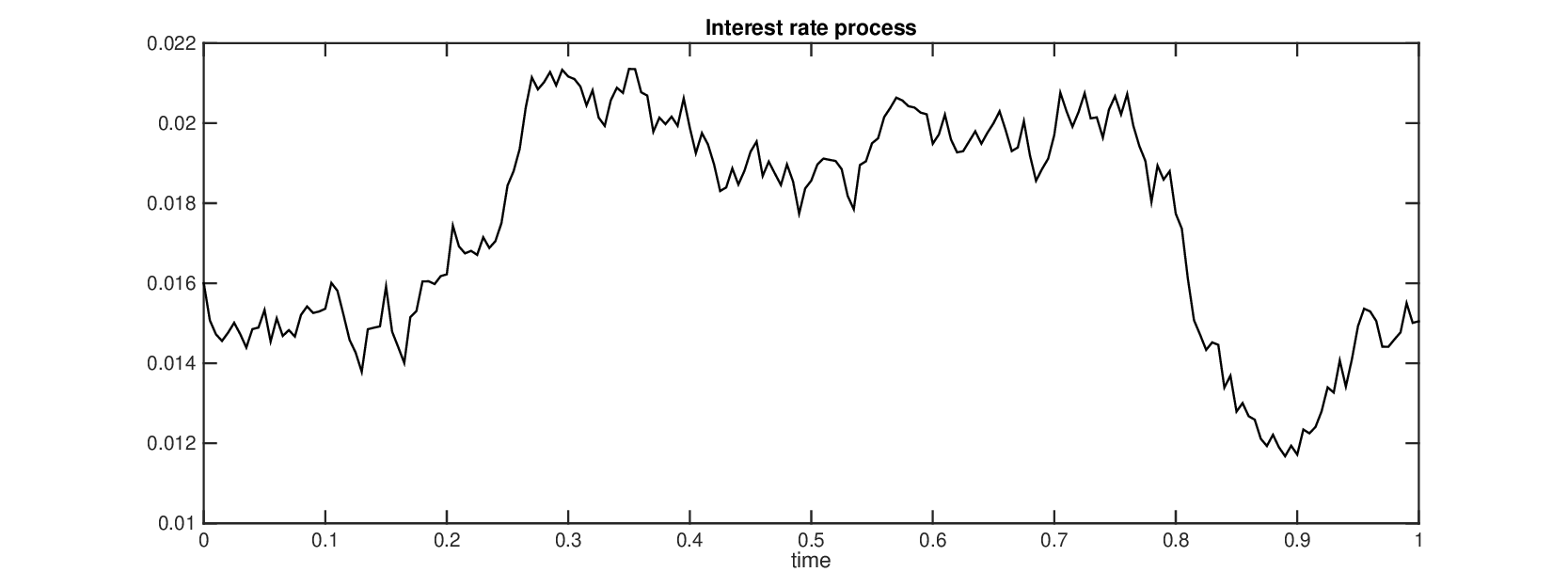}
	\caption{Interest rate process}
	\label{fig:IRProcess}
\end{figure}
\begin{figure}[!]
	\centering
	\includegraphics[width=0.7\linewidth]{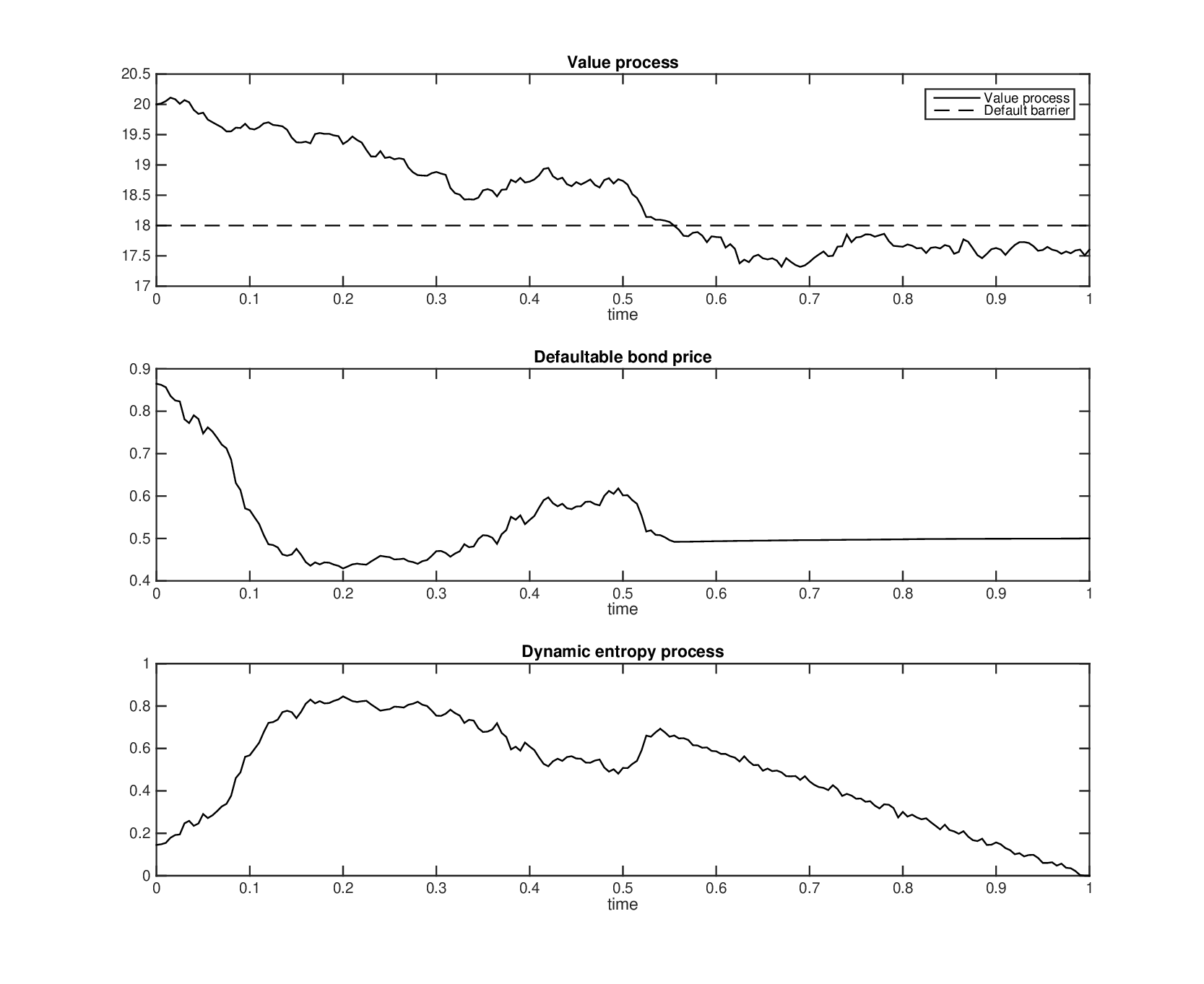}
	\caption{Realization with default}
	\label{fig:DefautableBondPriceWithDefault}
\end{figure}
\begin{figure}[!]
	\centering
	\includegraphics[width=0.7\linewidth]{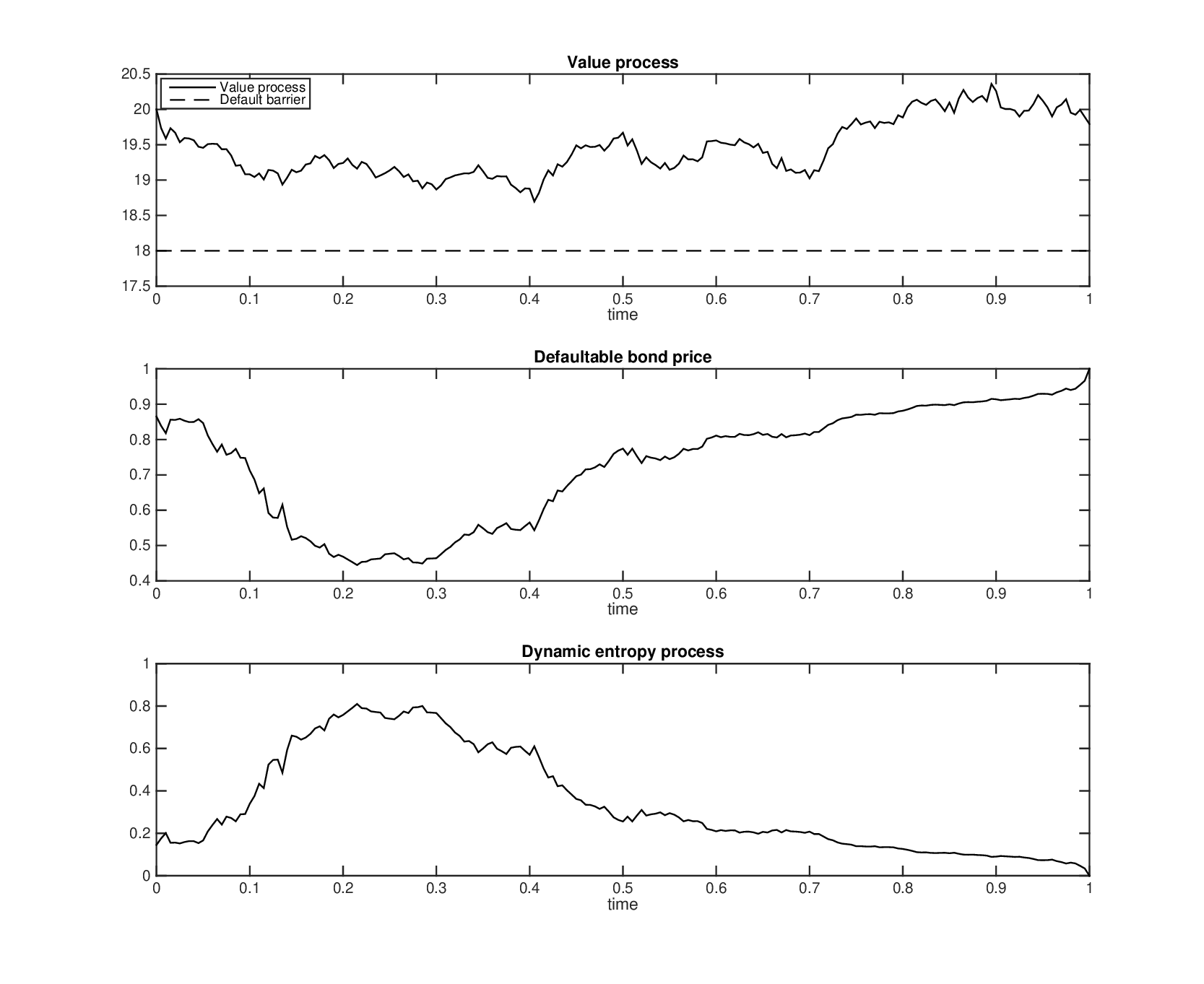}
	\caption{Realization without default}
	\label{fig:DefautableBondPriceWithoutDefault}
\end{figure}

\section{Conclusions}\label{conclusion}
In this paper we introduced and solved the entropic measure transform problem in Proposition~\ref{proposition:duality:relation:a}.  The general characterization of the density process of the optimal measure $\mathbb{P}^{\star}$, was characterized by a semimartignale BSDE in Theorem~\ref{thrm_main}.  Proposition~\ref{proposition:duality:relation:a} was used to interpret the conditional entropic risk of a random variable as a penalized conditional expectation, under the optimal measure $\mathbb{P}^{\star}$.  Proposition~\ref{thrm_affine_characterization} used the EMT to obtain a convenient formula for computing the conditional expectation of a process which can be expressed as an affine process under a related measure.  
 
Theorem~\ref{thrm_main} and Proposition~\ref{proposition:duality:relation:a} were then used to characterized the pricing problem for default-free bonds from a new perspective by formulating an entropic measure transform problem.  The solution of these problems consists of the entropic measure transform and the value process and these are characterized by the solution of a decoupled nonlinear FBSDE.  The explicit solutions to FBSDEs under ATSMs and QTSMs can be found in \citet{hyndman2009forward} and \citet{Zhou2010}. We provide an equivalence relationship between the optimal control approach in \citet{gombani2012arbitrage} and the entropic measure transform approach. Since the EMT approach, introduced in this paper, can easily incorporate jumps, it is more flexible then the OSC approach.  We also extend the EMT problem to include jumps. We give explicit solutions to the related FBSDEs with jumps, which generalizes \citet{hyndman2009forward} and \citet{Zhou2010}. Finally we formulate the EMT problem for defaultable bonds, in which case the related FBSDE generally does not have a completely explicit solution due to the dependence on the general specification of the default time and recovery amount of the random terminal value of the BSDE. However, the partially explicit solution still simplifies the problem of solving the nonlinear FBSDE.   Future research can consider specific models for the default time and recovery scheme where it may be possible to give more explicit solutions.  It is also of interest to consider the EMT problem in the context of pricing problems for other derivative securities.

\appendix

\section{Appendix}
This appendix discusses technical results on Riccati equations.  
\subsection{Riccati equations}
\begin{proposition}
	The following decoupled Riccati equations admit a pair of unique explicit solutions. 
	\begin{align}
	&\dot{q}_s+q_sA+A^\prime q_s+\frac{1}{2}(q_s^\prime+q_s)\Sigma\Sigma^\prime(q_s^\prime+q_s)-Q=0_{n\times
	n}, \quad s\in[0,T]\label{eq:Riccati:QTSM:q}\\
	&\dot{u}_s+u_sA+B^\prime(q_s^\prime+q_s)+u_s\Sigma\Sigma^\prime(q_s^\prime+q_s)-R^\prime=0_{1\times n},
	\quad s\in[0,T]\label{eq:Riccati:QTSM:u}\\
	&q_T=0_{n\times n},\quad u_T=0_{1\times n}. \label{eq:Riccati:QTSM:boundary}
	\end{align}
\end{proposition}
\begin{proof}
	We first prove equation \eqref{eq:Riccati:QTSM:q} admits a unique explicit solution. By taking the transpose of both sides of equation \eqref{eq:Riccati:QTSM:q} we find
	\begin{equation}\label{eq:Riccati:QTSM:proof:a}
	\dot{q}^\prime_s+A^\prime q_s^\prime+q_s^\prime A +\frac{1}{2}(q_s^\prime+q_s)\Sigma\Sigma^\prime(q_s^\prime+q_s)-Q=0_{n\times n}, 
	\end{equation}
	Adding equation \eqref{eq:Riccati:QTSM:proof:a} to equation \eqref{eq:Riccati:QTSM:q} gives 
	\begin{equation}\label{eq:Riccati:QTSM:proof:b}
	(\dot{q}_s+\dot{q}^\prime_s)+A^\prime (q_s+q_s^\prime)+(q_s+q_s^\prime )A +(q_s^\prime+q_s)\Sigma\Sigma^\prime(q_s^\prime+q_s)-2Q=0_{n\times n}, 
	\end{equation}
	and subtracting equation \eqref{eq:Riccati:QTSM:proof:a} from equation \eqref{eq:Riccati:QTSM:q} to find
	\begin{equation}\label{eq:Riccati:QTSM:proof:c}
	(\dot{q}_s-\dot{q}^\prime_s)+A^\prime (q_s-q_s^\prime)+(q_s-q_s^\prime )A=0.
	\end{equation}
	Define 
	$$U_s=\frac{q_s^\prime+q_s}{2},\quad V_s=\frac{q_s-q_s^\prime}{2},$$
	and by the terminal condition \eqref{eq:Riccati:QTSM:boundary} we have
	$$U_T=0_{n\times n},\quad V_T=0.$$
	Hence $U_s$ and $V_s$ satisfy the following equations
	\begin{align}
	&\dot{U}_s+A^\prime U_s+U_sA +U_s\Sigma\Sigma^\prime U_s-Q=0_{n\times n}, \label{eq:Riccati:QTSM:proof:d}\\
	&\dot{V}_s+A^\prime V_s+V_sA=0_{n\times n},\\
	& U_T=0_{n\times n},\quad V_T=0_{n\times n}. \label{eq:Riccati:QTSM:proof:e}
	\end{align}
	By \citet[Theorem B.1]{gombani2012arbitrage} there exists a pair of unique $(U_s,V_s)$ satisfying
	equations \eqref{eq:Riccati:QTSM:proof:d}-\eqref{eq:Riccati:QTSM:proof:e}. Moreover, we actually have $V_s=0_{n\times n}$ which means
	$q_s=q^\prime_s$, so $q_s$ is symmetric, and $q_s=U_s$.
	
	After we obtain the solution $q_s$, equation \eqref{eq:Riccati:QTSM:u} is simplified as an ODE for $u_s$, which can be solved explicitly as in \citet[Corollary B.3]{gombani2012arbitrage}.
	
\end{proof}

\bibliographystyle{abbrvnat}
\bibliography{mybib}

\begin{thebibliography}{29}
\providecommand{\natexlab}[1]{#1}
\providecommand{\url}[1]{\texttt{#1}}
\expandafter\ifx\csname urlstyle\endcsname\relax
  \providecommand{\doi}[1]{doi: #1}\else
  \providecommand{\doi}{doi: \begingroup \urlstyle{rm}\Url}\fi

\bibitem[Altman et~al.(2004)Altman, Resti, and Sironi]{altman2004default}
E.~Altman, A.~Resti, and A.~Sironi.
\newblock Default recovery rates in credit risk modelling: a review of the
  literature and empirical evidence.
\newblock \emph{Economic Notes}, 33\penalty0 (2):\penalty0 183--208, 2004.

\bibitem[Bielecki and Rutkowski(2002)]{bielecki2002credit}
T.~R. Bielecki and M.~Rutkowski.
\newblock \emph{Credit Risk: Modeling, Valuation and Hedging}.
\newblock Springer-Verlag, Berlin-Heidelberg-New York, 2002.

\bibitem[Bj{\"o}rk(2004)]{bjork2004arbitrage}
T.~Bj{\"o}rk.
\newblock \emph{Arbitrage Theory in Continuous Time}.
\newblock Oxford University Press, Oxford, 2004.

\bibitem[Bj{\"o}rk and Land{\'e}n(2002)]{landen2002term}
T.~Bj{\"o}rk and C.~Land{\'e}n.
\newblock On the term structure of futures and forward prices.
\newblock In \emph{Mathematical Finance-Bachelier Congress 2000 (Paris)},
  Springer Finance, pages 111--149. Springer, Berlin, 2002.

\bibitem[Cohen and Elliott(2015)]{SamElliotStochastics}
S.~N. Cohen and R.~J. Elliott.
\newblock \emph{Stochastic calculus and applications}.
\newblock Probability and its Applications. Springer, Cham, second edition,
  2015.

\bibitem[Cox et~al.(1985)Cox, Ingersoll~{Jr}, and Ross]{cox1985theory}
J.~C. Cox, J.~E. Ingersoll~{Jr}, and S.~A. Ross.
\newblock A theory of the term structure of interest rates.
\newblock \emph{Econometrica}, pages 385--407, 1985.

\bibitem[Cuchiero(2011)]{cuchiero2011affine}
C.~Cuchiero.
\newblock \emph{Affine and polynomial processes}.
\newblock PhD thesis, ETH Zurich, 2011.

\bibitem[Dai~Pra et~al.(1996)Dai~Pra, Meneghini, and
  Runggaldier]{dai1996connections}
P.~Dai~Pra, L.~Meneghini, and W.~J. Runggaldier.
\newblock Connections between stochastic control and dynamic games.
\newblock \emph{Mathematics of Control, Signals and Systems}, 9\penalty0
  (4):\penalty0 303--326, 1996.

\bibitem[Delong(2013)]{delong2013backward}
{\L}.~Delong.
\newblock \emph{Backward Stochastic Differential Equations with Jumps and Their
  Actuarial and Financial Applications}.
\newblock Springer, London, 2013.

\bibitem[Detlefsen and Scandolo(2005)]{DCRMs}
K.~Detlefsen and G.~Scandolo.
\newblock Conditional and dynamic convex risk measures.
\newblock \emph{Finance Stoch.}, 9\penalty0 (4):\penalty0 539--561, 2005.

\bibitem[Duffie and Kan(1996)]{duffie1996yield}
D.~Duffie and R.~Kan.
\newblock A yield-factor model of interest rates.
\newblock \emph{Mathematical Finance}, 6\penalty0 (4):\penalty0 379--406, 1996.

\bibitem[Duffie and Singleton(1999)]{duffie1999modeling}
D.~Duffie and K.~J. Singleton.
\newblock Modeling term structures of defaultable bonds.
\newblock \emph{Review of Financial Studies}, 12\penalty0 (4):\penalty0
  687--720, 1999.

\bibitem[Duffie et~al.(2003)Duffie, Filipovi{\'c}, and
  Schachermayer]{duffie2003affine}
D.~Duffie, D.~Filipovi{\'c}, and W.~Schachermayer.
\newblock Affine processes and applications in finance.
\newblock \emph{Annals of Applied Probability}, 13\penalty0 (3):\penalty0
  984--1053, 2003.

\bibitem[Elliott and van~der Hoek(2001)]{Elliott}
R.~J. Elliott and J.~van~der Hoek.
\newblock Stochastic flows and the forward measure.
\newblock \emph{Finance and Stochastics}, 5:\penalty0 511--525, 2001.

\bibitem[Gobet et~al.(2005)Gobet, Lemor, and Warin]{gobet2005regression}
E.~Gobet, J.-P. Lemor, and X.~Warin.
\newblock A regression-based {M}onte {C}arlo method to solve backward
  stochastic differential equations.
\newblock \emph{Ann. Appl. Probab.}, 15\penalty0 (3):\penalty0 2172--2202,
  2005.

\bibitem[Gombani and Runggaldier(2013)]{gombani2012arbitrage}
A.~Gombani and W.~J. Runggaldier.
\newblock Arbitrage-free multifactor term structure models: A theory based on
  stochastic control.
\newblock \emph{Mathematical Finance}, 23\penalty0 (4):\penalty0 659--686,
  2013.

\bibitem[{Gonon} and {Teichmann}(2018)]{2018arXiv180107796G}
L.~{Gonon} and J.~{Teichmann}.
\newblock {Linearized Filtering of Affine Processes Using Stochastic Riccati
  Equations}.
\newblock \emph{arXiv e-prints}, page arXiv:1801.07796, 2018.

\bibitem[Hull and White(1994)]{hull1994numerical}
J.~C. Hull and A.~D. White.
\newblock Numerical procedures for implementing term structure models {II}:
  Two-factor models.
\newblock \emph{The Journal of Derivatives}, 2\penalty0 (2):\penalty0 37--48,
  1994.

\bibitem[Hyndman(2007)]{hyndman2007forward}
C.~B. Hyndman.
\newblock Forward--backward {SDE}s and the {CIR} model.
\newblock \emph{Statistics $\&$ Probability Letters}, 77\penalty0
  (17):\penalty0 1676--1682, 2007.

\bibitem[Hyndman(2009)]{hyndman2009forward}
C.~B. Hyndman.
\newblock A forward--backward {SDE} approach to affine models.
\newblock \emph{Mathematics and Financial Economics}, 2\penalty0 (2):\penalty0
  107--128, 2009.

\bibitem[Hyndman and Zhou(2015)]{Zhou2010}
C.~B. Hyndman and X.~Zhou.
\newblock Explicit solutions of quadratic {FBSDE}s arising from quadratic term
  structure models.
\newblock \emph{Stochastic Analysis and Applications}, 33\penalty0
  (3):\penalty0 464--492, 2015.

\bibitem[Jacod and Shiryaev(2003)]{JacodShiryaevLimitTheorems}
J.~Jacod and A.~N. Shiryaev.
\newblock \emph{Limit theorems for stochastic processes}, volume 288 of
  \emph{Grundlehren der Mathematischen Wissenschaften [Fundamental Principles
  of Mathematical Sciences]}.
\newblock Springer-Verlag, Berlin, second edition, 2003.

\bibitem[Jacod and Yor(1977)]{JacodYorThrm}
J.~Jacod and M.~Yor.
\newblock \'{E}tude des solutions extr\'{e}males et repr\'{e}sentation
  int\'{e}grale des solutions pour certains probl\`emes de martingales.
\newblock \emph{Z. Wahrscheinlichkeitstheorie und Verw. Gebiete}, 38\penalty0
  (2):\penalty0 83--125, 1977.

\bibitem[Keller-Ressel and Mayerhofer(2015)]{KellerEMAFF}
M.~Keller-Ressel and E.~Mayerhofer.
\newblock Exponential moments of affine processes.
\newblock \emph{Ann. Appl. Probab.}, 25\penalty0 (2), 2015.

\bibitem[Kobylanski(2000)]{kobylanski2000backward}
M.~Kobylanski.
\newblock Backward stochastic differential equations and partial differential
  equations with quadratic growth.
\newblock \emph{Annals of Probability}, 28\penalty0 (2):\penalty0 558--602,
  2000.

\bibitem[Longstaff and Schwartz(1992)]{longstaff1992interest}
F.~A. Longstaff and E.~S. Schwartz.
\newblock Interest rate volatility and the term structure: A two-factor general
  equilibrium model.
\newblock \emph{Journal of Finance}, 47\penalty0 (4):\penalty0 1259--1282,
  1992.

\bibitem[Ma and Yong(1999)]{ma1999forward}
J.~Ma and J.~Yong.
\newblock \emph{Forward-Backward Stochastic Differential Equations and Their
  Applications}, volume 1702 of \emph{Lecture Notes in Mathematics}.
\newblock Springer-Verlag, Berlin, 1999.

\bibitem[Merton(1974)]{merton1974pricing}
R.~C. Merton.
\newblock On the pricing of corporate debt: The risk structure of interest
  rates.
\newblock \emph{Journal of Finance}, 29\penalty0 (2):\penalty0 449--470, 1974.

\bibitem[Va{\v{s}}{\'{\i}}{\v{c}}ek(1977)]{vasicek1977equilibrium}
O.~Va{\v{s}}{\'{\i}}{\v{c}}ek.
\newblock An equilibrium characterization of the term structure.
\newblock \emph{Journal of Financial Economics}, 5\penalty0 (2):\penalty0
  177--188, 1977.

\end{thebibliography}
\end{document}